\documentclass[11pt,a4paper]{article}

\usepackage{amsmath,amssymb,amsthm}
\usepackage{amsfonts}
\usepackage{mathtools}
\usepackage{booktabs}
\usepackage{multirow}
\usepackage{enumitem}
\usepackage{array}
\usepackage[margin=2.5cm]{geometry}
\usepackage[colorlinks=true,linkcolor=blue,citecolor=red,urlcolor=blue]{hyperref}
\allowdisplaybreaks

\newtheorem{theorem}{Theorem}[section]
\newtheorem{corollary}[theorem]{Corollary}
\newtheorem{lemma}[theorem]{Lemma}
\newtheorem{proposition}[theorem]{Proposition}
\theoremstyle{definition}
\newtheorem{definition}[theorem]{Definition}
\newtheorem{remark}[theorem]{Remark}
\newtheorem{example}[theorem]{Example}

\newcommand{\Z}{\mathbb{Z}}
\newcommand{\C}{{\mathcal C}}
\newcommand{\cH}{{\mathcal H}}
\newcommand{\cB}{{\mathcal B}}
\newcommand{\cM}{{\mathcal M}}
\newcommand{\cR}{{\mathcal R}}
\newcommand{\cV}{{\mathcal V}}
\newcommand{\cZ}{{\mathcal Z}}
\newcommand{\zero}{{\mathbf{0}}}
\newcommand{\one}{{\mathbf{1}}}
\newcommand{\two}{{\mathbf{2}}}
\newcommand{\four}{{\mathbf{4}}}
\newcommand{\uu}{\mathbf{u}}
\newcommand{\vv}{\mathbf{v}}
\newcommand{\ww}{\mathbf{w}}
\newcommand{\xx}{\mathbf{x}}
\newcommand{\yy}{\mathbf{y}}
\newcommand{\zz}{\mathbf{z}}
\newcommand{\bfe}{\mathbf{e}}
\newcommand{\wt}{{\rm wt}}
\newcommand{\rank}{\operatorname{rank}}
\newcommand{\kernel}{\operatorname{ker}}
\newcommand{\ord}{\operatorname{o}}
\newcommand{\K}{\operatorname{K}}
\newcommand{\Span}{\operatorname{span}}
\newcommand{\sym}{\operatorname{sym}}

\newcommand{\sel}{\operatorname{sel}}

\begin{document}

\title{\bf The kernel-block rank profiles of the $\mathbb{Z}_2\mathbb{Z}_4$-linear and the
$\mathbb{Z}_{2^s}$-linear Hadamard codes, and a complete classification of the
$\mathbb{Z}_2\mathbb{Z}_4\mathbb{Z}_8$-linear Hadamard codes\thanks{This work has been granted by the Juan de la Cierva 2024 grant (JDC2024-053082-I), funded by Ministerio de Ciencia, Innovación y Universidades (MICIU/AEI/10.13039/501100011033) and co-funded by the European Social Fund Plus (FSE+).}}

\author{
Dipak K. Bhunia\\
\small Departament de Matem\`atiques,\\ \small Universitat Polit\`ecnica de Catalunya,\\
\small Barcelona, Spain
}
\date{}

\maketitle

\begin{abstract}
The $\mathbb{Z}_2\mathbb{Z}_4\mathbb{Z}_8$-additive codes are subgroups of $\mathbb{Z}_2^{\alpha_1} \times \mathbb{Z}_4^{\alpha_2} \times \mathbb{Z}_8^{\alpha_3}$.  Similarly, the $\Z_{2^s}$-additive codes of length $n$ are subgroups of $\Z_{2^s}^n$. A $\mathbb{Z}_2\mathbb{Z}_4\mathbb{Z}_8$-linear or $\Z_{2^s}$-linear Hadamard code is a Hadamard code which is the Gray map image of a $\mathbb{Z}_2\mathbb{Z}_4\mathbb{Z}_8$-additive or $\Z_{2^s}$-additive code, respectively. The kernel of a binary code containing the zero word partitions the set of binary coordinates into blocks, two coordinates lying in the same block when every kernel word takes the same value in both of them, and the
\emph{kernel-block rank profile} of the code is the multiset of the dimensions of its linear span punctured on those blocks.  For the family $H^{t_1,t_2,t_3}$ of $\mathbb{Z}_2\mathbb{Z}_4\mathbb{Z}_8$-linear Hadamard codes, this invariant is known explicitly and gives a complete classification of the family. In this paper, we compute it for the $\mathbb{Z}_2\mathbb{Z}_4$-linear and $\mathbb{Z}_{2^s}$-linear Hadamard families with which those codes are compared, and we prove that it is constant for every nonlinear $\mathbb{Z}_2\mathbb{Z}_4$-linear Hadamard code and for every nonlinear $\mathbb{Z}_{2^s}$-linear Hadamard code $\bar H^{a_1,\dots,a_s}$ with $s\geq2$.  The second statement is obtained without any rank formula, by exhibiting coordinate permutations that preserve the code and act transitively on its kernel blocks, which makes the argument uniform in $s$.  We also prove a descent theorem: if $\bar H^{a_1,\dots,a_s}$ is nonlinear and $a_1\geq2$, then the code punctured on one of its kernel blocks is the $\mathbb{Z}_{2^{s-1}}$-linear Hadamard code $\bar H^{a_1,\dots,a_{s-1}}$. Consequently, the constant local rank equals $\rank(\bar H^{a_1,\dots,a_{s-1}})$ and is at least $t-\kappa+2$, where $2^t$ is the length of the code and $\kappa$ the dimension of the kernel. These results separate every nonlinear $\mathbb{Z}_2\mathbb{Z}_4\mathbb{Z}_8$-linear Hadamard code from every $\mathbb{Z}_4$-linear, $\mathbb{Z}_2\mathbb{Z}_4$-linear and $\mathbb{Z}_{2^s}$-linear Hadamard code of the same length, for every $s\geq2$, except for the single infinite family $H^{1,1,t-4}$ and $\bar H^{2,0,t-5}$ with $t\geq5$. The members of this infinite family agree in the rank, the dimension of the kernel, and the kernel-block rank profile, but a two-block refinement of the profile separates them.
\end{abstract}



\section{Introduction}\label{sec:intro}

Let $\Z_{2^s}$ be the ring of integers modulo $2^s$ with $s\geq1$.  The set of
$n$-tuples over $\Z_{2^s}$ is denoted by $\Z_{2^s}^n$, and in this paper the
elements of $\Z_{2^s}^n$ are also called vectors.  A code over $\Z_2$ of
length $N$ is a nonempty subset of $\Z_2^N$, and it is linear if it is a
subspace of $\Z_2^N$.  Similarly, a nonempty subset of $\Z_{2^s}^n$ is a
$\Z_{2^s}$-additive code if it is a subgroup of the additive group of
$\Z_{2^s}^n$.  A $\Z_2\Z_4\Z_8$-additive code is a subgroup of
$\Z_2^{\alpha_1}\times\Z_4^{\alpha_2}\times\Z_8^{\alpha_3}$.  Note that a
$\Z_2\Z_4\Z_8$-additive code is a linear code over $\Z_2$ when
$\alpha_2=\alpha_3=0$, a $\Z_4$-additive or $\Z_8$-additive code when
$\alpha_1=\alpha_3=0$ or $\alpha_1=\alpha_2=0$, respectively, and a
$\Z_2\Z_4$-additive code when $\alpha_3=0$.  The order of a vector
$\uu\in\Z_{2^s}^n$, denoted by $\ord(\uu)$, is the smallest positive integer $m$
such that $m\uu=(0,\dots,0)$.  Also, the order of a vector
$\uu\in\Z_2^{\alpha_1}\times\Z_4^{\alpha_2}\times\Z_8^{\alpha_3}$, denoted by
$\ord(\uu)$, is the smallest positive integer $m$ such that
$m\uu=(0,\dots,0\mid 0,\dots,0\mid 0,\dots,0)$.
Throughout the paper, vectors are written in boldface, and $\zero$ and $\one$
denote the all-zero and the all-one vector, the length being always clear from
the context.

Two binary codes $C_1$ and $C_2$ of length $N$ are said to be equivalent if
there are a vector $\uu\in\Z_2^N$ and a permutation of coordinates $\pi$ such
that $C_2=\{\uu+\pi(\xx):\xx\in C_1\}$. The Hamming weight of a vector
$\uu\in\Z_2^N$, denoted by $\wt_H(\uu)$, is the number of its nonzero
coordinates, and the Hamming distance $d_H(\uu,\vv)$ of two vectors is the
number of coordinates in which they differ; hence, $d_H(\uu,\vv)
=\wt_H(\uu-\vv)$.  The minimum distance of a code $C$ over $\Z_2$ is
$d(C)=\min\{d_H(\uu,\vv):\uu,\vv\in C,\ \uu\not=\vv\}$. 

In \cite{Sole}, a Gray map from $\Z_4$ to $\Z_2^2$ is defined as
$\phi(0)=(0,0)$, $\phi(1)=(0,1)$, $\phi(2)=(1,1)$ and $\phi(3)=(1,0)$.  There
exist different generalizations of this Gray map, which go from $\Z_{2^s}$ to
$\Z_2^{2^{s-1}}$ \cite{Carlet,Codes2k,dougherty,Nechaev,Krotov:2007}.  In this
paper, we focus on Carlet's Gray map \cite{Carlet}, from $\Z_{2^s}$ to
$\Z_2^{2^{s-1}}$, which is a particular case of the one given in
\cite{Krotov:2007,ShiKrotov2019} satisfying
$\sum\lambda_i\phi_s(2^i)=\phi_s(\sum\lambda_i2^i)$ \cite{KernelZ2s}.
Specifically,
\begin{equation}\label{eq:Carlet}
\phi_s(u)=(u_{s-1},u_{s-1},\dots,u_{s-1})+(u_0,\dots,u_{s-2})Y_{s-1},
\end{equation}
where $u\in\Z_{2^s}$; $[u_0,u_1,\dots,u_{s-1}]_2$ is the binary expansion of
$u$, that is, $u=\sum_{i=0}^{s-1}u_i2^i$ with $u_i\in\{0,1\}$; and $Y_{s-1}$
is a matrix of size $(s-1)\times2^{s-1}$ whose columns are all the vectors of
$\Z_2^{s-1}$, ordered in ascending order by reading them as binary expansions
of the elements of $\Z_{2^{s-1}}$.  Note that $\phi_1$ is the identity map and
that $\phi_2$ is the Gray map $\phi$ given above.  We define
$\Phi_s:\Z_{2^s}^n\rightarrow\Z_2^{n2^{s-1}}$ as the component-wise extended
map of $\phi_s$, and a Gray map $\Phi$ from
$\Z_2^{\alpha_1}\times\Z_4^{\alpha_2}\times\Z_8^{\alpha_3}$ to $\Z_2^N$, where
$N=\alpha_1+2\alpha_2+4\alpha_3$, by
$$\Phi(\uu_1\mid \uu_2\mid \uu_3)=(\uu_1,\Phi_2(\uu_2),\Phi_3(\uu_3)),$$ for any $\uu_i\in\Z_{2^i}^{\alpha_i}$, where $1\leq i\leq3$. 

Let $\C\subseteq\Z_{2^s}^n$ be a $\Z_{2^s}$-additive code of length $n$.  We
say that the Gray map image of $\C$, say $C=\Phi_s(\C)$, is a
$\Z_{2^s}$-linear code of length $n2^{s-1}$.  Since $\C$ is a subgroup of
$\Z_{2^s}^n$, it is isomorphic to
$\Z_{2^s}^{t_1}\times\Z_{2^{s-1}}^{t_2}\times\dots\times\Z_2^{t_s}$, and we
say that $\C$, or equivalently $C=\Phi_s(\C)$, is of type
$(n;t_1,\dots,t_s)$.  Similarly, if
$\C\subseteq\Z_2^{\alpha_1}\times\Z_4^{\alpha_2}\times\Z_8^{\alpha_3}$ is a
$\Z_2\Z_4\Z_8$-additive code, we say that its Gray map image $C=\Phi(\C)$ is a
$\Z_2\Z_4\Z_8$-linear code of length $\alpha_1+2\alpha_2+4\alpha_3$.  Since
$\C$ can be seen as a subgroup of $\Z_8^{\alpha_1+\alpha_2+\alpha_3}$, it is
isomorphic to $\Z_8^{t_1}\times\Z_4^{t_2}\times\Z_2^{t_3}$, and we say that
$\C$, or equivalently $C=\Phi(\C)$, is of type
$(\alpha_1,\alpha_2,\alpha_3;t_1,t_2,t_3)$.  Note that a $\Z_2\Z_4$-linear
code $\C$ \cite{ccsg,BookZ2Z4} can be seen as a $\Z_2\Z_4\Z_8$-linear code of
type $(\alpha_1,\alpha_2,0;0,t_2,t_3)$.  In this case, we also say that the
type of $\C$ is directly $(\alpha_1,\alpha_2;t_2,t_3)$.  Unlike linear codes
over finite fields, linear codes over rings do not have a basis, but there
exists a generator matrix for these codes having minimum number of rows.  If
$\C$ is a $\Z_2\Z_4\Z_8$-additive code of type
$(\alpha_1,\alpha_2,\alpha_3;t_1,t_2,t_3)$, then $|\C|=8^{t_1}4^{t_2}2^{t_3}$
and there exists a generator matrix with $t_1+t_2+t_3$ rows.

A binary code of length $N$, $2N$ codewords and minimum distance $N/2$ is
called a Hadamard code.  Hadamard codes can be constructed from Hadamard
matrices \cite{Key,WMcwill}.  Note that linear Hadamard codes are in fact
first order Reed--Muller codes, or equivalently, the dual of extended Hamming
codes \cite{WMcwill}.  It is also important to note that Hadamard codes are
two weight codes, which have been widely studied in \cite{ShiTwoHomWeight,TwoWeightSole}.  Most Hadamard codes are, however,
nonlinear, and their classification is still an open problem.  One fruitful
way of attacking it is to realise some of them as Gray map images of additive
codes over rings or over mixed alphabets and then to decide which of the
resulting codes are equivalent, each such result giving a partial
classification of nonlinear Hadamard codes.  The $\Z_4$ representation of the
Kerdock, Preparata and Goethals codes \cite{Sole} was the starting point of
this point of view, and codes over $\Z_{p^s}$ go back to Blake \cite{Blake}
and Shankar \cite{Shankar}.

From a more practical point of view, since Hadamard codes are optimal and have
a high correction capability, they appear in different aspects related to the
transmission of information, such as in digital communication with satellites
\cite{H07}, in CDMA phones to modulate the transmission of information and
minimize interference with other transmissions \cite{SBT98} and, in general,
in different OCDMA multiple access systems to allow access to multiple users
asynchronously and simultaneously \cite{HYT04}.  Other applications are found
in cryptography \cite{Nyb91} or in information hiding (steganography and
watermarking) \cite{YLL03}.  See \cite{H07} for more applications in other
fields.

Two structural properties of codes over $\Z_2$ are the rank and the dimension
of the kernel.  The rank of a code $C$ over $\Z_2$ is simply the dimension of
the linear span, $\langle C\rangle$, of $C$.  The kernel of a code $C$ over
$\Z_2$ is defined as $\K(C)=\{\xx\in\Z_2^N:\xx+C=C\}$ \cite{BGH83}.  If the
all-zero vector belongs to $C$, then $\K(C)$ is a linear subcode of $C$.  Note
also that if $C$ is linear, then $\K(C)=C=\langle C\rangle$.  We denote the
rank of $C$ as $\rank(C)$ and the dimension of the kernel as $\kernel(C)$.
Both are equivalence invariants, so two codes with different pairs
$(r,k)=(\rank(C),\kernel(C))$ are nonequivalent.  The converse is false in
general, and it is precisely the failure of the converse that makes
classification results difficult.  In \cite{Z2Z4Z8Profile}, a new invariant
was introduced to repair that failure for one family of Hadamard codes, and
the present paper carries the same programme through for the families with
which those codes have to be compared.

The $\Z_{2^s}$-additive codes such that after the Gray map $\Phi_s$ give
Hadamard codes are called $\Z_{2^s}$-additive Hadamard codes and the
corresponding images are called $\Z_{2^s}$-linear Hadamard codes.  Similarly,
the $\Z_2\Z_4\Z_8$-additive codes such that after the Gray map $\Phi$ give
Hadamard codes are called $\Z_2\Z_4\Z_8$-additive Hadamard codes and the
corresponding images are called $\Z_2\Z_4\Z_8$-linear Hadamard codes.  It is
known that $\Z_4$-linear Hadamard codes, that is, $\Z_2\Z_4$-linear Hadamard
codes with $\alpha_1=0$, and $\Z_2\Z_4$-linear Hadamard codes with
$\alpha_1\not=0$ can be classified by using either the rank or the dimension
of the kernel \cite{Kro:2001:Z4_Had_Perf,PRV06}.  Moreover, in \cite{KV2015},
it is shown that each $\Z_4$-linear Hadamard code is equivalent to a
$\Z_2\Z_4$-linear Hadamard code with $\alpha_1\not=0$.  Later, in
\cite{KernelZ2s,HadamardZps,EquivZ2s,ZpsEquivalance}, a recursive construction
for $\Z_{p^s}$-linear Hadamard codes, with $p$ prime, is described, the
linearity is established, and a partial classification by using the dimension
of the kernel is obtained, giving the exact amount of nonequivalent such codes
for some parameters.  In \cite{fernandez2019mathbb}, a complete classification
of $\Z_8$-linear Hadamard codes by using the rank and dimension of the kernel
is provided, giving the exact amount of nonequivalent such codes. For
$s\geq4$, however, no formula for the rank is known, and therefore that
invariant is not even available there.  Finally, a recursive construction of
$\Z_2\Z_4\Z_8$-additive Hadamard codes $\cH^{t_1,t_2,t_3}$, with all the
$\alpha_i$ nonzero, was given in \cite{Z2Z4Z8Construction}. The linearity
and kernel of the corresponding $\Z_2\Z_4\Z_8$-linear Hadamard codes were
determined in \cite{Z2Z4Z8Linearity}, and their rank was determined in
\cite{Z2Z4Z8Rank}; we write $H^{t_1,t_2,t_3}=\Phi(\cH^{t_1,t_2,t_3})$ for the corresponding binary codes. For any $t\geq2$, the full classification of $\Z_p\Z_{p^2}$-linear generalized Hadamard codes of length $p^t$, with $\alpha_1\not=0$, $\alpha_2\not=0$, and $p\geq3$ prime, is given in \cite{ZpZp2Construction,ZpZp2Classification}, by using just the
dimension of the kernel.

The invariant of \cite{Z2Z4Z8Profile} is the following.  The kernel of a binary
code $C$ with $\zero\in C$ is a linear subcode, so it can be used to compare
coordinate positions.  Two binary coordinates are declared equivalent when
every kernel word takes the same value in both of them, and the classes of the resulting partition are called the \emph{kernel blocks} of $C$.  On each block $B$, we puncture the linear span and record the \emph{local rank} $\rho_C(B)=\dim\langle C\rangle|_B$, and the multiset
$$
\cR(C)=\{\!\{\rho_C(B):\ B\text{ a kernel block of }C\}\!\}
$$
is the \emph{kernel-block rank profile} of $C$.  It is an equivalence
invariant.  In \cite{Z2Z4Z8Profile}, the profile of $H^{t_1,t_2,t_3}$ was computed:
the kernel partition has $2^{t_1+t_2+t_3-1}$ blocks, all of size
$2^{2t_1+t_2}$, and the profile takes the two values $t_2+\binom{t_1+2}{2}$
and $t_2+2+\binom{t_1+1}{2}$, whose difference is $t_1-1$.  This classified the
family internally: two codes of the family with the same length are equivalent
if and only if their types coincide.

The present paper completes the classification by comparing those codes with
the families that were already known, and our results are of three kinds.

First, we compute the profile of the two comparison families.  Proposition
\ref{prop:z2z4-profile} shows that a nonlinear $\Z_2\Z_4$-linear Hadamard code
$H^{U,V}$ has $2^{U+V-1}$ kernel blocks, all of size $2^{U}$, and constant
profile with value $U+1$.  Theorem \ref{thm:z2s-profile-constant} shows that
the profile of every nonlinear $\Z_{2^s}$-linear Hadamard code $\bar H^{a_1,\dots,a_s}$ is constant as well, for every $s\geq2$.  Its proof computes no local rank: it produces two
families of coordinate permutations that map the code onto itself, one
translating the additive coordinates and one shifting the Gray coordinates,
and shows that together they carry any kernel block onto any other one.  No
rank formula is used, which is why the argument works uniformly in $s$,
including the range $s\geq4$ where no such formula exists.

Second, we identify the constant value in a form stronger than a numerical
formula.  Theorem \ref{thm:descent} states that, when $\bar H^{a_1,\dots,a_s}$
is nonlinear and $a_1\geq2$, the code punctured on one of its kernel blocks is
again a Hadamard code of the same kind, one alphabet lower, namely
$\bar H^{a_1,\dots,a_{s-1}}$.  The common local rank is therefore
$\rank(\bar H^{a_1,\dots,a_{s-1}})$, and Corollary
\ref{cor:local-rank-bound} gives the sharp bound
$\rho_{a_1,\dots,a_s}\geq t-\kappa+2$, with equality if and only if $s=2$, or
$s=3$ and $a_1=2$, where $2^t$ is the length of the code and $\kappa$ the dimension of the
kernel.

Third, we draw the classification.  Since the profile of $H^{t_1,t_2,t_3}$ is
nonconstant precisely when $t_1\geq2$, every such code is separated at once
from every $\Z_4$-linear, $\Z_2\Z_4$-linear and $\Z_{2^s}$-linear Hadamard
code of the same length.  When $t_1=1$, the profile is constant and equal to
the minimum value $t-\kappa+2$, so only the two extremal situations of
Corollary \ref{cor:local-rank-bound} remain, and the rank settles them except
for the single infinite family
$$
H^{1,1,t-4}\quad\text{and}\quad \bar H^{2,0,t-5},\qquad t\geq5,
$$
whose members share the length, the rank, the dimension of the kernel and the
whole profile.  For this family, we introduce a two-block refinement of the
profile, the maximum $\mu(C)$ of $\dim\langle C\rangle|_{B\cup B'}$ over the
pairs of distinct blocks, and compute $\mu=6$ on one side and $\mu=7$ on the
other.  The outcome is Theorem \ref{thm:cross-family}: no nonlinear
$\Z_2\Z_4\Z_8$-linear Hadamard code is equivalent to a $\Z_4$-linear, a
$\Z_2\Z_4$-linear or a $\Z_{2^s}$-linear Hadamard code of the same length, for
any $s\geq2$.  This is exactly the statement left as further research in
\cite[\S 6]{Z2Z4Z8Linearity}, where it had been established only for
$5\leq t\leq11$ and with the help of {\sc Magma} \cite{Magma}.

The paper is organised as follows.  Section \ref{sec:preliminaries} recalls
the kernel-block rank profile and its basic properties from \cite{Z2Z4Z8Profile},
the facts about $H^{t_1,t_2,t_3}$ that we need, and the two comparison
families.  Section \ref{sec:z2z4-profile} computes the profile of the
$\Z_2\Z_4$-linear Hadamard codes, and Section \ref{sec:z2s-profile} proves
that the profile of every nonlinear $\Z_{2^s}$-linear Hadamard code is
constant.  Section \ref{sec:descent} identifies the punctured code on a kernel
block and derives the local rank together with its lower bound.  Section
\ref{sec:cross-family} contains the classification, the two-block invariant
and the comparison tables, and Section \ref{sec:conclusions} summarises the
results and indicates further research.

\section{Preliminaries}\label{sec:preliminaries}

This section collects everything that is used later. 
Subsections \ref{subsec:profile-recalled} and \ref{subsec:mixed-recalled}
recall from \cite{Z2Z4Z8Profile,Z2Z4Z8Construction,Z2Z4Z8Linearity,Z2Z4Z8Rank} the
kernel-block rank profile and the facts about the $\Z_2\Z_4\Z_8$-linear
Hadamard codes, while Subsections
\ref{subsec:z2z4-family} and \ref{subsec:z2s-family} describe the two
families with which the $\Z_2\Z_4\Z_8$-linear Hadamard family will be compared. 

Throughout the paper, a code with two superscripts is $\Z_2\Z_4$-linear, one
with three superscripts and no bar is $\Z_2\Z_4\Z_8$-linear, and one with a
bar is $\Z_{2^s}$-linear.

\subsection{Puncturing and the kernel-block rank profile}
\label{subsec:profile-recalled}

We first fix the action of a permutation on vectors, since the bookkeeping
matters below.  A permutation $\pi$ of the set $\{1,\dots,N\}$ of coordinate
positions acts on $\Z_2^N$ by
\begin{equation}\label{eq:action}
\pi(\xx)_{\pi(p)}=x_p\qquad\text{for every }p\in\{1,\dots,N\},
\end{equation}
that is, the entry of $\xx$ at the position $p$ becomes the entry of
$\pi(\xx)$ at the position $\pi(p)$. The action is additive, so
$\pi(\xx+\yy)=\pi(\xx)+\pi(\yy)$ for all $\xx,\yy\in\Z_2^N$.

We will constantly restrict codes and linear spaces to selected coordinate
positions.  For $\xx\in\Z_2^N$ and $B\subseteq\{1,\dots,N\}$, we write
$\xx|_B$ for the vector obtained from $\xx$ by deleting the coordinates whose
index lies outside $B$, and, for a set $D\subseteq\Z_2^N$, we put
$D|_B=\{\xx|_B:\xx\in D\}$. This operation is called puncturing outside $B$.
Since deleting coordinates is a $\Z_2$-linear map, $D|_B$ is a linear space whenever
$D$ is one, and
\begin{equation}\label{eq:span-puncture}
\langle D\rangle|_B=\langle D|_B\rangle
\qquad\text{for every } D\subseteq\Z_2^N .
\end{equation}

The two notions interact as follows.  Let $S$ be a set of coordinate positions, and let
$\pi|_{S}$ denote the restriction of $\pi$ to $S$, a bijection from $S$ onto
$\pi(S)$. By \eqref{eq:action}, $\pi(\xx)|_{\pi(S)}$ is
obtained from $\xx|_{S}$ by relabelling the positions through $\pi|_{S}$.
Hence, for every linear space $V\subseteq\Z_2^N$, $\pi(V)|_{\pi(S)}$ is obtained 
from $V|_{S}$ by a coordinate permutation. In particular,
\begin{equation}\label{eq:relabel}
\dim\pi(V)|_{\pi(S)}=\dim V|_{S},
\end{equation}
for every linear space  $V\subseteq\Z_2^N$ and every  $S$ .

\begin{definition}[\cite{Z2Z4Z8Profile}]\label{def:kernel-block-profile}
Let $C\subseteq\Z_2^N$ be a binary code with $\zero\in C$ and put $K=\K(C)$.
Define a relation on the set $\{1,\dots,N\}$ of binary coordinates by
\begin{equation}\label{eq:kernel-relation}
i\sim_K j\quad\Longleftrightarrow\quad
x_i=x_j\ \text{ for every }\xx=(x_1,\dots,x_N)\in K .
\end{equation}
This is an equivalence relation; its classes are called the \emph{kernel
blocks} of $C$, and the corresponding partition of $\{1,\dots,N\}$ is denoted
by $\cB(K)$.  For a block $B\in\cB(K)$, we call
$\rho_C(B)=\dim\langle C\rangle|_B$ the \emph{local rank} of $C$ on $B$, and
the multiset
\begin{equation}\label{eq:profile}
\cR(C)=\{\!\{\rho_C(B):\ B\in\cB(K)\}\!\}
\end{equation}
is called the \emph{kernel-block rank profile} of $C$.
\end{definition}

Two conventions are attached to this definition and will be used without
further comment.  First, since $K$ is a linear space, the condition in
\eqref{eq:kernel-relation} may be checked on any basis of $K$.
Second, \eqref{eq:profile} is a multiset and not a set: the blocks carry no
preferred order, but the number of blocks having each local rank is part of
the invariant, and we write $\{\!\{\rho^{[M]}\}\!\}$ when the value $\rho$
occurs with multiplicity $M$.  We say that the profile is \emph{constant} when
all its entries are equal.

The next three statements are the basic properties of the invariant, proved in
\cite{Z2Z4Z8Profile}. 

\begin{lemma}[\cite{Z2Z4Z8Profile}]\label{lem:profile-invariant}
Let $C,D\subseteq\Z_2^N$ be equivalent binary codes with $\zero\in C$ and
$\zero\in D$, say $D=\yy+\pi(C)$ for a coordinate permutation $\pi$ and a
vector $\yy\in\Z_2^N$.  Then, $\pi$ maps $\cB(\K(C))$ bijectively onto
$\cB(\K(D))$, and $\rho_D(\pi(B))=\rho_C(B)$ for every $B\in\cB(\K(C))$.  In
particular, $\cR(C)=\cR(D)$.
\end{lemma}

\begin{lemma}[\cite{Z2Z4Z8Profile}]\label{lem:linear-profile}
If $C$ is a linear Hadamard code of length $2^t$, then all its kernel blocks
are singletons and $\cR(C)=\{\!\{1^{[2^t]}\}\!\}$.
\end{lemma}

\begin{lemma}[\cite{Z2Z4Z8Profile}]\label{lem:plotkin-blocks}
Let $C\subseteq\Z_2^N$ be a binary code with $\zero\in C$ and $\one\in\K(C)$,
and let
\begin{equation}\label{eq:plotkin-extension}
C^{+}=\{(\xx,\xx):\xx\in C\}\cup\{(\xx,\xx+\one):\xx\in C\}
\end{equation}
be its Plotkin extension.  Then,
\begin{equation}\label{eq:plotkin-kernel-span}
\begin{aligned}
\K(C^{+})&=\{(\xx,\xx)+\varepsilon(\zero,\one):\xx\in\K(C),\ \varepsilon\in\Z_2\},\\
\langle C^{+}\rangle&=\{(\zz,\zz)+\varepsilon(\zero,\one):
\zz\in\langle C\rangle,\ \varepsilon\in\Z_2\}.
\end{aligned}   
\end{equation}
Every kernel block $B$ of $C$ gives rise to exactly two kernel blocks $B^{0}$
and $B^{1}$ of $C^{+}$, one inside each copy, and
$$
|B^{0}|=|B^{1}|=|B|,\qquad
\rho_{C^{+}}(B^{0})=\rho_{C^{+}}(B^{1})=\rho_{C}(B).
$$
In particular, $\cR(C^{+})$ is obtained from $\cR(C)$ by doubling all the
multiplicities, and $\one\in\K(C^{+})$.
\end{lemma}

Finally, we recall from \cite{Z2Z4Z8Profile} the elementary lemma describing the
effect of an equivalence on the span and on the kernel, which we use when the
translation vector has to be tracked and not only the two dimensions.

\begin{lemma}[\cite{Z2Z4Z8Profile}]\label{lem:rank-kernel-invariance}
Let $C,D\subseteq\Z_2^N$ be binary codes with $\zero\in C$ and $\zero\in D$,
and suppose that $D=\yy+\pi(C)$ for some coordinate permutation $\pi$ and some
$\yy\in\Z_2^N$.  Then, $\langle D\rangle=\pi(\langle C\rangle)$ and
$\K(D)=\pi(\K(C))$.  In particular, $\rank(D)=\rank(C)$ and
$\kernel(D)=\kernel(C)$.
\end{lemma}

\subsection{Boolean functions attached to binary coordinates}
\label{subsec:boolean}

Every rank computed in this paper is the dimension of a space of Boolean
functions, and the translation is provided by the following three lemmas from
\cite{Z2Z4Z8Rank, Z2Z4Z8Profile}.

Let $\C$ be an additive code, over $\Z_{2^s}$ or over a mixed alphabet, let
$\cM$ be its set of messages, that is, the set of tuples of coefficients of a
fixed generating set, and let $\Phi$ be the corresponding Gray map.  For a
binary coordinate $\ell$ of $\Phi(\C)$, the \emph{binary coordinate function}
$f_\ell:\cM\rightarrow\Z_2$ sends a message to the value of the corresponding
binary codeword at $\ell$.  Writing every message coefficient in binary
identifies $\cM$ with a space $\Z_2^{m}$, so each $f_\ell$ is a Boolean
function of the binary message digits.

\begin{lemma}[\cite{Z2Z4Z8Rank}]\label{lem:evaluation-principle}
Let $\Omega$ be a finite set, let $f_1,\dots,f_N:\Omega\rightarrow\Z_2$. Fix
an ordering $\omega_1,\dots,\omega_{|\Omega|}$ of $\Omega$, and let
$C_\Omega\subseteq\Z_2^N$ be the set of rows of the $|\Omega|\times N$ binary
matrix $\bigl(f_j(\omega_i)\bigr)_{i,j}$.  Then, for every
$B\subseteq\{1,\dots,N\}$,
\begin{equation}\label{eq:restricted-evaluation-rank}
\dim\langle C_\Omega\rangle|_B=\dim\Span\{f_j:\ j\in B\}.
\end{equation}
\end{lemma}

Applied with $\Omega=\cM$, this says that the local rank of $\Phi(\C)$ on a
set $B$ of binary coordinates is the dimension of the span of the coordinate
functions attached to the positions of $B$.

\begin{lemma}[{\cite{Z2Z4Z8Rank}}]\label{lem:gray-digits}
Let $\Omega$ be a finite set and let $u:\Omega\rightarrow\Z_{2^s}$ be any map.
For $0\leq i\leq s-1$, let $u_i:\Omega\rightarrow\Z_2$ send $\omega$ to the
$i$th binary digit of $u(\omega)$, and, for $1\leq\ell\leq 2^{s-1}$, let
$g_\ell:\Omega\rightarrow\Z_2$ send $\omega$ to the $\ell$th entry of
$\phi_s(u(\omega))$.  Then,
$$
\Span\{g_1,\dots,g_{2^{s-1}}\}=\Span\{u_0,u_1,\dots,u_{s-1}\}.
$$
\end{lemma}

\begin{lemma}[\cite{Carlet,HadamardZps}]\label{lem:gray-shift}
Let $s\geq1$ and let $\lambda,w\in\Z_{2^s}$.  Then,
$\phi_s(2^{s-1}\lambda+w)=\phi_s(2^{s-1}\lambda)+\phi_s(w)$.
\end{lemma}

\begin{lemma}[\cite{Z2Z4Z8Profile}]\label{lem:mobius}
Let $\Omega$ be a finite set, let $\theta=(\theta_1,\dots,\theta_n)$ be free
binary parameters and suppose that, for every $\theta\in\Z_2^n$, a function
$G_\theta:\Omega\rightarrow\Z_2$ is given by
\begin{equation}\label{eq:parameter-polynomial}
G_\theta=\sum_{I\subseteq\{1,\dots,n\}}\Bigl(\prod_{i\in I}\theta_i\Bigr)f_I,
\end{equation}
where the functions $f_I:\Omega\rightarrow\Z_2$ do not depend on $\theta$.
Then, $\Span\{G_\theta:\theta\in\Z_2^n\}=\Span\{f_I:I\subseteq\{1,\dots,n\}\}$.
\end{lemma}

\begin{lemma}[\cite{Z2Z4Z8Rank}]\label{lem:carries}
Let $z_1,\dots,z_n\in\Z_2$ and let $w=z_1+\cdots+z_n$ be their ordinary
integer sum.  Then, for every $r\geq0$, the $r$th binary digit of $w$ equals
$\sym_{2^r}(z_1,\dots,z_n)$, where $\sym_j$ denotes the $j$th elementary
symmetric Boolean polynomial.  In particular, the bottom digit, the carry to
the next digit and the direct carry to the third digit are $\sym_1$, $\sym_2$
and $\sym_4$, respectively.
\end{lemma}

Every Boolean function has a unique algebraic normal form, so distinct
square-free monomials in independent binary variables are linearly
independent; this fact justifies every dimension count below.

\subsection{The $\Z_2\Z_4\Z_8$-linear Hadamard codes and their profile}
\label{subsec:mixed-recalled}

Let $\two$ and $\four$ denote the vectors having the element $2$, respectively
$4$, repeated in every coordinate.  For integers $t_1\geq1$, $t_2\geq0$ and
$t_3\geq1$, the matrix $A^{t_1,t_2,t_3}$ of \cite{Z2Z4Z8Construction} has
$t_1$ rows of order $8$, $t_2$ rows of order $4$ and $t_3$ rows of order $2$,
the first of which is always $\ww_0=(\one\mid\two\mid\four)$.  We write
$\cH^{t_1,t_2,t_3}$ for the additive code it generates and
$H^{t_1,t_2,t_3}=\Phi(\cH^{t_1,t_2,t_3})$ for its Gray image, which is a
binary Hadamard code of length $2^t$ with
\begin{equation}\label{eq:length-relation}
t+1=3t_1+2t_2+t_3 .
\end{equation}
When $t_3=1$, the additive coordinates are described explicitly in
\cite[Proposition 2.3]{Z2Z4Z8Rank}: after deleting from every column of $A^{t_1,t_2,1}$ its first entry, the $\Z_2$ part consists of all the vectors of $\Z_2^{t_1+t_2}$, the $\Z_4$ part of all the normalized primitive vectors of $\Z_4^{t_1+t_2}$, and the $\Z_8$ part of all the vectors $(q_1,\dots,q_{t_1},2r_1,\dots,2r_{t_2})$
with $(q_1,\dots,q_{t_1})$ normalized primitive in $\Z_8^{t_1}$ and
$(r_1,\dots,r_{t_2})\in\Z_4^{t_2}$ arbitrary, each one occurring exactly once.
Here, a vector over $\Z_{2^s}$ is \emph{primitive} when at least one of its
entries is odd, and \emph{normalized primitive} when moreover its first odd
entry equals $1$.  For $t_3>1$, the matrix is obtained from $A^{t_1,t_2,1}$ by
applying $t_3-1$ times the duplication that appends a row of order $2$.

\begin{theorem}[\cite{Z2Z4Z8Construction,Z2Z4Z8Linearity,Z2Z4Z8Rank}]
\label{thm:recalled-mixed}
Let $t_1\geq1$, $t_2\geq0$ and $t_3\geq1$.  Then, $H^{t_1,t_2,t_3}$ is linear
if and only if $(t_1,t_2)=(1,0)$. For the nonlinear ones,
$\K(H^{t_1,t_2,t_3})$ is the Gray image of the subcode of the elements of
order at most two; hence, a basis of it is
$\{\Phi(\tfrac{\ord(\ww)}{2}\ww)\}$ with $\ww$ running over the rows of
$A^{t_1,t_2,t_3}$, and
$
\kernel\bigl(H^{t_1,t_2,t_3}\bigr)=t_1+t_2+t_3 .
$
Moreover, for every admissible type,
\begin{equation}\label{eq:rank-binomial}
\rank(H^{t_1,t_2,t_3})=t_3-1+4t_1+4\binom{t_1}{2}+2\binom{t_1}{3}
+\binom{t_1}{4}+t_2\binom{t_1+2}{2}+\binom{t_2+1}{2}.
\end{equation}
\end{theorem}

For the computations of Section \ref{sec:cross-family}, we also need the
notation of \cite{Z2Z4Z8Profile} for the message digits and the resulting
description of the local spaces.  Fix $t_3=1$, write $\ww_1,\dots,\ww_{t_1}$
for the rows of order $8$ and $\vv_1,\dots,\vv_{t_2}$ for the rows of order
$4$ of $A^{t_1,t_2,1}$.  Then, every additive codeword is
$\epsilon\ww_0+\sum_ix_i\ww_i+\sum_jy_j\vv_j$ with $\epsilon\in\Z_2$,
$x_i\in\Z_8$ and $y_j\in\Z_4$. We expand the coefficients in binary as
\begin{equation}\label{eq:message-digits}
x_i=a_i+2b_i+4c_i,\qquad y_j=d_j+2e_j .
\end{equation}
The digits $a_i,b_i,c_i,d_j,e_j$ and $\epsilon$ are then treated as
independent binary variables.  For $S\subseteq\{1,\dots,t_1\}$ and
$R\subseteq\{1,\dots,t_2\}$, we abbreviate
\begin{equation}\label{eq:symmetric-notation}
\alpha_S=\sum_{i\in S}a_i,\quad \beta_S=\sum_{i\in S}b_i,\quad
\gamma_S=\sum_{i\in S}c_i,\quad \delta_R=\sum_{j\in R}d_j,\quad
\eta_R=\sum_{j\in R}e_j,
\end{equation}
and we write $\sym_2(a_S)$, $\sym_2(b_S)$, $\sym_2(d_R)$ and $\sym_4(a_S)$ for
the elementary symmetric polynomials of the corresponding sublists; for
instance, $\sym_2(a_S)=\sum_{\{i,k\}\subseteq S}a_ia_k$ and
$\sym_4(a_S)=\sum_{\{i,j,k,l\}\subseteq S}a_ia_ja_ka_l$.

The kernel blocks of $H^{t_1,t_2,1}$ are indexed in \cite{Z2Z4Z8Profile} by the
\emph{kernel labels}: the block $B_{u,v}$, with $(u,v)\in\Z_2^{t_1}\times
\Z_2^{t_2}$, consists of the binary coordinates at which the kernel generators
$\Phi(4\ww_i)$ and $\Phi(2\vv_j)$ take the values $u_i$ and $v_j$.  The next
proposition collects the outcome of that analysis; it is the input we need
here.

\begin{proposition}[\cite{Z2Z4Z8Profile}]\label{prop:recalled-local-bases}
Let $H^{t_1,t_2,1}$ be nonlinear.  Then, every $B_{u,v}$ is nonempty, the
kernel partition has $2^{t_1+t_2}$ blocks, all of size $2^{2t_1+t_2}$, and the
punctured span on a block is spanned by the following functions, which are
moreover linearly independent.
\begin{enumerate}[label=\textup{(\roman*)}]
\item On $B_{\zero,\zero}$:
\begin{equation}\label{eq:B00-basis}
\epsilon,\quad a_i\ (1\leq i\leq t_1),\quad d_j\ (1\leq j\leq t_2),\quad
b_i\ (1\leq i\leq t_1),\quad a_ia_k\ (1\leq i<k\leq t_1).
\end{equation}
\item On $B_{\zero,v}$ with $v\not=\zero$, writing $R=\{j:v_j=1\}$,
$F_v=\epsilon+\eta_R+\sym_2(d_R)$ and $g_i=b_i+a_i\delta_R$:
\begin{equation}\label{eq:B0v-basis}
a_i\ (1\leq i\leq t_1),\quad d_j\ (1\leq j\leq t_2),\quad F_v,\quad
g_i\ (1\leq i\leq t_1),\quad a_ia_k\ (1\leq i<k\leq t_1).
\end{equation}
\item On $B_{u,v}$ with $u\not=\zero$, writing $S=\{i:u_i=1\}$,
$R=\{j:v_j=1\}$, $i_0=\min S$, $\Psi_i=b_i+a_iQ_0$ and
\begin{equation}\label{eq:Q0F0}
\begin{aligned}
Q_0&=\beta_S+\delta_R+\sym_2(a_S),\\
F_0&=\epsilon+\gamma_S+\eta_R+\sym_2(b_S)+\sym_2(d_R)+\sym_4(a_S)
    +\beta_S\delta_R+\bigl(\beta_S+\delta_R\bigr)\sym_2(a_S):
\end{aligned}
\end{equation}
\begin{equation}\label{eq:Buv-basis}
a_i\ (1\leq i\leq t_1),\quad d_j\ (1\leq j\leq t_2),\quad Q_0,\quad F_0,
\quad \Psi_i\ (i\not=i_0),\quad a_ia_k\ (i,k\not=i_0,\ i<k).
\end{equation}
\end{enumerate}
\end{proposition}

\begin{theorem}[\cite{Z2Z4Z8Profile}]\label{thm:profile}
Let $H^{t_1,t_2,t_3}$ be nonlinear, that is, $(t_1,t_2)\not=(1,0)$, and put
\begin{equation}\label{eq:rho-def}
\rho_0=t_2+\binom{t_1+2}{2},
\qquad
\rho_1=t_2+2+\binom{t_1+1}{2}.
\end{equation}
Then, the kernel partition of $H^{t_1,t_2,t_3}$ has exactly
$2^{t_1+t_2+t_3-1}$ blocks, all of size $2^{2t_1+t_2}$, and
\begin{equation}\label{eq:profile-formula}
\cR\bigl(H^{t_1,t_2,t_3}\bigr)
=\bigl\{\!\bigl\{\ \rho_0^{\,[2^{t_2+t_3-1}]},\
\rho_1^{\,[(2^{t_1}-1)2^{t_2+t_3-1}]}\ \bigr\}\!\bigr\}.
\end{equation}
Moreover, $\rho_0-\rho_1=t_1-1$, so the profile is constant if and only if
$t_1=1$, in which case
\begin{equation}\label{eq:profile-t1-one}
\cR\bigl(H^{1,t_2,t_3}\bigr)
=\bigl\{\!\bigl\{(t_2+3)^{\,[2^{t_2+t_3}]}\bigr\}\!\bigr\}.
\end{equation}
\end{theorem}

Finally, we recall the list of the pairs, one member of which lies in the
$\Z_2\Z_4\Z_8$-linear Hadamard family and the other in one of the two comparison families, that the rank and the dimension of the kernel fail to separate.

\begin{proposition}[\cite{Z2Z4Z8Rank}]\label{prop:recalled-cross-collisions}
Let $H^{t_1,t_2,t_3}$ be nonlinear and of length $2^t$.
\begin{enumerate}[label=\textup{(\roman*)}]
\item If a $\Z_4$-linear or a $\Z_2\Z_4$-linear Hadamard code of the same
length has the same rank and the same dimension of the kernel, then
$(t_1,t_2)=(2,1)$ and that code is permutation equivalent to $H^{5,t-9}$, with
$t\geq10$; conversely, $H^{2,1,t-7}$ and $H^{5,t-9}$ do share all three
parameters.
\item If a nonlinear $\Z_8$-linear Hadamard code $\bar H$ of the same length
has the same rank and the same dimension of the kernel, then the pair
$(H^{t_1,t_2,t_3},\bar H)$ is one of the following five, and each of the five
does share all three parameters:
\begin{equation}\label{eq:z8-five}
\begin{array}{llll}
\bigl(H^{1,1,t-4},\ \bar H^{2,0,t-5}\bigr), & t\geq5, &
(r,k)=(t+3,\,t-2),\\[1mm]
\bigl(H^{2,1,t-7},\ \bar H^{1,5,t-12}\bigr), & t\geq12, &
(r,k)=(t+11,\,t-4),\\[1mm]
\bigl(H^{1,5,t-12},\ \bar H^{4,0,t-11}\bigr), & t\geq13, &
(r,k)=(t+21,\,t-6),\\[1mm]
\bigl(H^{2,6,t-17},\ \bar H^{4,3,t-17}\bigr), & t\geq18, &
(r,k)=(t+51,\,t-9),\\[1mm]
\bigl(H^{4,5,t-21},\ \bar H^{6,2,t-21}\bigr), & t\geq22, &
(r,k)=(t+117,\,t-12).
\end{array}
\end{equation}
\end{enumerate}
\end{proposition}

\subsection{The $\Z_2\Z_4$-linear Hadamard codes}\label{subsec:z2z4-family}

Deleting the $\Z_8$ part specialises the construction of Subsection
\ref{subsec:mixed-recalled} to the recursive construction of the
$\Z_2\Z_4$-additive Hadamard codes of type $(\alpha_1,\alpha_2;U,V)$ with
$\alpha_1\not=0$, $\alpha_2\not=0$, $U\geq 1$ and $V\geq 1$ given in \cite{PRV06,RSV08}. For the
general theory of $\Z_2\Z_4$-additive codes, we refer to \cite{ccsg,BookZ2Z4}.
Deleting that part merges the two extension steps that add a row of order $8$
and a row of order $4$ into a single step adding a row of order $4$, and the
remaining step becomes the duplication that adds a row of order $2$. Indeed, we start from 
\begin{equation*}
A^{1,1}=\left(\begin{array}{cc|c}
1&1&2\\
0&1&1
\end{array}\right)
\end{equation*} and then apply the two constructions below. 
If $A^{\ell-1,1}=(A_1\mid A_2)$ with $\ell\geq2$ is already been constructed, we form
\begin{equation}\label{eq:z2z4-construction-order4}
A^{\ell,1}=\left(\begin{array}{cc|ccccc}
A_1&A_1&M_1&A_2&A_2&A_2&A_2\\
\zero&\one&\one&\zero&\one&\two&\mathbf{3}
\end{array}\right),
\end{equation}
where $M_1=2A_1$ up to a permutation of columns. The construction \eqref{eq:z2z4-construction-order4} is applied 
until $\ell=U$. If $A^{U,\ell-1}=(A_1\mid A_2)$ with $\ell\geq2$ is already constructed, we form
\begin{equation}\label{eq:z2z4-construction-order2}
A^{U,\ell}=\left(\begin{array}{cc|cc}
A_1&A_1&A_2&A_2\\
\zero&\one&\zero&\two
\end{array}\right),
\end{equation} and the construction \eqref{eq:z2z4-construction-order2} is applied until
$\ell=V$. Thus, in this way, we obtain $A^{U, V}$.

Every $\Z_2\Z_4$-additive Hadamard code
with $\alpha_1\not=0$ and $\alpha_2\not=0$ is permutation equivalent to one of
the codes so obtained \cite{PRV06,RSV08}, so it is enough to work with these.
We write $\cH^{U,V}$ for the code generated by $A^{U,V}$ and
$H^{U,V}=\Phi(\cH^{U,V})$ for its Gray image, whose binary length is $2^t$
with
\begin{equation}\label{eq:z2z4-length}
t+1=2U+V .
\end{equation}
The description of the additive coordinates specialises as well: after
deleting the first entry of every column of $A^{U,1}$, the $\Z_2$ part
consists of all the vectors of $\Z_2^{U}$ and the $\Z_4$ part of all the
normalized primitive vectors of $\Z_4^{U}$, and for $V>1$, the matrix is
obtained by applying the duplication exactly $V-1$ times.  By \cite{PRV06},
the code $H^{U,V}$ is linear if and only if $U=1$ and, for the nonlinear ones,
\begin{equation}\label{eq:z2z4-invariants}
\kernel\bigl(H^{U,V}\bigr)=U+V,
\qquad
\rank\bigl(H^{U,V}\bigr)=V+2U+\binom{U}{2}.
\end{equation}
Also by \cite{PRV06}, the kernel of a nonlinear $H^{U,V}$ is the Gray map
image of the subcode of the elements of order at most two, exactly as in
Theorem \ref{thm:recalled-mixed}.

Finally, by \cite{KV2015}, every $\Z_4$-linear Hadamard code is equivalent to
a $\Z_2\Z_4$-linear Hadamard code with $\alpha_1\not=0$ and $\alpha_2\not=0$.
In all the comparisons below it therefore suffices to consider the family
$H^{U,V}$ with $U\geq2$ and $V\geq1$, the $\Z_4$-linear Hadamard codes being
automatically covered.  We write the parameters of a $\Z_2\Z_4$-linear
Hadamard code as $U$ and $V$, and never as $t_2$ and $t_3$, in order to keep
them apart from the parameters of the $\Z_2\Z_4\Z_8$-linear Hadamard codes.

\subsection{The $\Z_{2^s}$-linear Hadamard codes}\label{subsec:z2s-family}

We recall the construction of \cite{KernelZ2s}.  Let $s\geq2$ and put
\begin{equation}\label{eq:Ti}
T_i=\bigl\{j\cdot2^{i-1}:\ 0\leq j\leq 2^{s-i+1}-1\bigr\}
=2^{i-1}\Z_{2^s},\qquad 1\leq i\leq s.
\end{equation}
Thus, $T_1=\Z_{2^s}$, and each $T_i$ is an ideal of $\Z_{2^s}$ whose elements
have order at most $2^{s-i+1}$.  Let $a_1,\dots,a_s$ be nonnegative integers
with $a_1\geq1$, put $m=a_1+\cdots+a_s$, and let $\bar A^{a_1,\dots,a_s}$ be
the matrix with $m$ rows whose columns are exactly all the vectors $\zz^T$
with
\begin{equation}\label{eq:z2s-columns}
\zz\in\{1\}\times T_1^{a_1-1}\times T_2^{a_2}\times\cdots\times T_s^{a_s}.
\end{equation}
We write $\ww_1,\dots,\ww_m$ for the rows, and we call \emph{level} of
$\ww_i$ the index $j$ such that $\ww_i$ is one of the $a_j$ rows whose entries
range over $T_j$ in \eqref{eq:z2s-columns}. Thus, the first $a_1$ rows have
level $1$, the next $a_2$ rows have level $2$, and so on.  A row $\ww_i$ of
level $j$ has $\ord(\ww_i)=2^{\,s-j+1}$, because its entries run over the
whole of $T_j$ and, in particular, include $2^{\,j-1}$, whose order in
$\Z_{2^s}$ is $2^{\,s-j+1}$.  By construction
$\ww_1=\one$, so the additive code $\bar\cH^{a_1,\dots,a_s}$ generated by
$\bar A^{a_1,\dots,a_s}$ contains every constant word $\lambda\one$ with
$\lambda\in\Z_{2^s}$.  We put
$\bar H^{a_1,\dots,a_s}=\Phi_s(\bar\cH^{a_1,\dots,a_s})$.  Every
$\Z_{2^s}$-additive Hadamard code is permutation equivalent to one of the
codes $\bar\cH^{a_1,\dots,a_s}$ \cite{KernelZ2s}, so it is enough to work
with these.
By \cite{KernelZ2s}, $\bar H^{a_1,\dots,a_s}$ is a binary Hadamard
code of length $2^t$ with
\begin{equation}\label{eq:z2s-length}
t+1=\sum_{i=1}^{s}(s-i+1)a_i.
\end{equation}
Define
\begin{equation}\label{eq:sigma}
\sigma=\begin{cases}
1,&\text{if }a_1\geq2,\\
s,&\text{if }a_1=1\text{ and }a_2=\cdots=a_s=0,\\
\min\{i:\ a_i>0,\ 2\leq i\leq s\},&\text{otherwise.}
\end{cases}
\end{equation}
The code $\bar H^{a_1,\dots,a_s}$ is linear if and only if
$(a_1,\dots,a_s)$ is either $(1,0,\dots,0,a_s)$ or $(1,0,\dots,0,1,a_s)$ when
$s\geq3$, by \cite{KernelZ2s}, and if and only if $a_1\leq2$ when
$s=2$, by \cite{Kro:2001:Z4_Had_Perf}. Moreover, by
\cite[Theorem 3]{KernelZ2s}, the nonlinear ones satisfy
\begin{equation}\label{eq:z2s-kernel}
\kernel\bigl(\bar H^{a_1,\dots,a_s}\bigr)=\sigma+\sum_{i=1}^{s}a_i .
\end{equation}

\begin{proposition}[{\cite{KernelZ2s}}]
\label{prop:z2s-kernel-basis}
Let $\bar H^{a_1,\dots,a_s}$ be nonlinear.  Then, a basis of
$\K(\bar H^{a_1,\dots,a_s})$ is
\begin{equation}\label{eq:z2s-kernel-basis}
\Phi_s\Bigl(\tfrac{\ord(\ww_i)}{2}\ww_i\Bigr)\ (1\leq i\leq m),
\qquad
\Phi_s\bigl(2^{\nu}\one\bigr)\ (0\leq\nu\leq\sigma-2),
\qquad
\Phi_s\bigl((2^{s-1}-1)\one\bigr),
\end{equation}
the middle family being empty when $\sigma=1$.
\end{proposition}

The first $m$ vectors are the Gray images of the elements of order two of a
generating set. The last
$\sigma$ are Gray images of constant words, and they are what makes the
$\Z_{2^s}$ case differ from the $\Z_2\Z_4\Z_8$ one.  For $s=3$, the rank was determined
in \cite{fernandez2019mathbb}:
\begin{equation}\label{eq:z8-rank}
\rank\bigl(\bar H^{a,b,c}\bigr)
=\frac{a^4-2a^3+35a^2+14a}{24}+\frac{b}{2}\bigl(a^2+a+b+1\bigr)+c+1 .
\end{equation}
No formula for the rank is known for $s\geq4$, which is exactly why the
comparison of \cite{Z2Z4Z8Rank} stops at $\Z_8$.

The families obtained for different values of $s$ are not disjoint, and the
overlap is described by the following theorem, which raises the alphabet by
$\ell$ levels at the cost of one unit of $a_1$ and $\ell$ units of $a_s$.

\begin{theorem}[{\cite{EquivZ2s}}]\label{thm:recalled-equivalence}
Let $s\geq2$ and let $\bar H^{a_1,\dots,a_s}$ be a $\Z_{2^s}$-linear Hadamard
code with $a_s\geq1$.  Then, for every $\ell$ with $1\leq\ell\leq a_s$, the
code $\bar H^{a_1,\dots,a_s}$ is permutation equivalent to the
$\Z_{2^{s+\ell}}$-linear Hadamard code
$$
\bar H^{\,1,\,0^{\ell-1},\,a_1-1,\,a_2,\dots,a_{s-1},\,a_s-\ell},
$$
where $0^{\ell-1}$ stands for a string of $\ell-1$ zeros.
\end{theorem}

Read in the opposite direction, Theorem \ref{thm:recalled-equivalence} allows us to lower
the alphabet whenever the first parameter equals $1$, and this is how we will
use it in Remark \ref{rem:chain-reduction}.

It is convenient to index the binary coordinates of $\bar H^{a_1,\dots,a_s}$
by pairs.  Let $\cZ$ be the set of additive coordinates of
$\bar\cH^{a_1,\dots,a_s}$, which we identify with the set of columns
\eqref{eq:z2s-columns}.  By \eqref{eq:Carlet}, the $2^{s-1}$ binary
coordinates produced by an additive coordinate $\zz\in\cZ$ are naturally
indexed by the vectors $h\in\Z_2^{s-1}$, and we write $(\zz,h)$ for the
corresponding binary coordinate.  Let $\xx=(x_1,\dots,x_m)$ be a message, let
$\uu\in\bar\cH^{a_1,\dots,a_s}$ be the corresponding additive codeword and let
\begin{equation}\label{eq:z2s-symbol}
W_{\xx}(\zz)=\sum_{i=1}^{m}x_i z_i\ \in\Z_{2^s}
\end{equation}
be its entry at $\zz$.  Then, $\Phi_s(\uu)$ takes at $(\zz,h)$ the value
\begin{equation}\label{eq:z2s-coordinate-value}
\phi_s\bigl(W_{\xx}(\zz)\bigr)_h
=W_{\xx}(\zz)_{s-1}+\sum_{i=0}^{s-2}h_i\,W_{\xx}(\zz)_i ,
\end{equation}
where $W_{\xx}(\zz)_0,\dots,W_{\xx}(\zz)_{s-1}$ are the binary digits of
$W_{\xx}(\zz)$.  When the message is fixed and no confusion can arise, we
abbreviate $W=W_{\xx}(\zz)$.  We will also use the following elementary
observation.

\begin{lemma}\label{lem:sigma-bound}
If $\bar H^{a_1,\dots,a_s}$ is nonlinear, then $1\leq\sigma\leq s-1$.
\end{lemma}

\begin{proof}
If $a_1\geq2$, then $\sigma=1$ by \eqref{eq:sigma} and $s\geq2$, so
$\sigma\leq s-1$.  Assume now $a_1=1$.  If $a_i=0$ for every $i$ with
$2\leq i\leq s-1$, then $(a_1,\dots,a_s)=(1,0,\dots,0,a_s)$ and, by the
linearity criterion recalled above, the code would be linear, which is
excluded by hypothesis.  Hence, $a_i>0$ for at least one index $i$ with
$2\leq i\leq s-1$, and the third case of \eqref{eq:sigma} gives
$\sigma=\min\{i:a_i>0,\ 2\leq i\leq s\}\leq s-1$.
\end{proof}
\section{The profile of the $\Z_2\Z_4$-linear Hadamard codes}
\label{sec:z2z4-profile}

For the $\Z_2\Z_4$-linear Hadamard codes, the profile can be computed directly
by the same method as for the $\Z_2\Z_4\Z_8$-linear Hadamard codes. We determine the
kernel labels, count the coordinates carrying each label, and compute the span
of the corresponding coordinate functions.  The computation is shorter here, because
there are no rows of order $8$.

\begin{proposition}\label{prop:z2z4-profile}
Let $H^{U,V}$ be a nonlinear $\Z_2\Z_4$-linear Hadamard code, that is, with
$U\geq2$ and $V\geq1$.  Then, its kernel partition has $2^{U+V-1}$ blocks, all
of size $2^{U}$, and
\begin{equation}\label{eq:z2z4-profile}
\cR\bigl(H^{U,V}\bigr)
=\bigl\{\!\bigl\{(U+1)^{\,[2^{U+V-1}]}\bigr\}\!\bigr\}.
\end{equation}
In particular, the profile is constant.
\end{proposition}

\begin{proof}
We first treat the case $V=1$ and then pass to a general $V$ by duplication.

Assume therefore that $V=1$, and let
$\ww_0=(\one\mid\two)$ be the distinguished row of order $2$, while
$\ww_1,\dots,\ww_{U}$ are the rows of order $4$ of $A^{U,1}$.  By
\cite{PRV06}, recalled at the end of Subsection \ref{subsec:z2z4-family}, the kernel of a nonlinear $H^{U,V}$ is the Gray map image of
the subcode of the elements of order at most two. Therefore, $\Phi(\ww_0)$ and
the $\Phi(2\ww_j)$ form a basis of the kernel.  Since $\phi_2(2)=(1,1)$,
we have $\Phi(\ww_0)=\Phi(\one\mid\two)=\one$, which takes the value $1$ everywhere
and so separates no two coordinates.  As in Subsection
\ref{subsec:mixed-recalled}, we may therefore attach to each binary coordinate
the vector of the values taken there by the remaining $U$ generators, and call
it the label of that coordinate. By \eqref{eq:kernel-relation}, two
coordinates lie in the same kernel block exactly when they have the same
label.

Let us compute the labels.  At a binary coordinate produced by an additive
coordinate of the $\Z_2$ part, every entry of the column lies in $\Z_2$, so
the entry of $2\ww_j$ there is $2$ times an element of $\Z_2$, hence zero; the
label is $\zero\in\Z_2^{U}$.  At a binary coordinate produced by the additive
coordinate of the $\Z_4$ part whose column is $(2,s_1,\dots,s_{U})^T$, the row
$\ww_j$ has the entry $s_j\in\Z_4$, so $2\ww_j$ has the entry $2s_j$ modulo
$4$, which is $0$ when $s_j$ is even and $2$ when $s_j$ is odd. Since
 $\phi_2(2)=(1,1)$, both binary coordinates produced
there receive the same value $s_j\bmod2$.  The label is therefore the parity
vector of $(s_1,\dots,s_{U})$, and it is nonzero, because that vector is
normalized primitive and hence has at least one odd entry.  Consequently, the
label $\zero$ occurs exactly on the $\Z_2$ part, every nonzero label occurs
only in the $\Z_4$ part, and, as we check next while counting, every label
does occur.

We now determine the sizes of the blocks.  By Subsection
\ref{subsec:z2z4-family}, the $\Z_2$ part of $A^{U,1}$ consists of all the
vectors of $\Z_2^{U}$, each occurring once, so it has $2^{U}$ additive
coordinates, each producing one binary coordinate. Hence, the block of label
$\zero$ has $2^{U}$ elements.  Fix now a nonzero label and let
$R\not=\emptyset$ be the corresponding set of odd positions, with
$j_0=\min R$.  The additive coordinates of this block are those whose column
is $(2,s_1,\dots,s_U)^T$ with $s_j$ odd exactly for $j\in R$.  Writing $R'$
for the set of positions at which the higher binary digit of $s_j$ equals $1$,
normalization means that the first odd entry, namely $s_{j_0}$, equals $1$,
that is, $j_0\notin R'$; equivalently
$R'\subseteq\{1,\dots,U\}\setminus\{j_0\}$; hence, one higher digit is
prescribed and the remaining $U-1$ are free.  There are therefore $2^{U-1}$
such additive coordinates, each producing two binary ones, and the block has
$2\cdot2^{U-1}=2^{U}$ elements.  Hence, all the $2^{U}$ labels occur, and all the
blocks have $2^{U}$ elements; as a check, $2^{U}\cdot2^{U}=2^{2U}=2^{t}$ by
\eqref{eq:z2z4-length} with $V=1$, which is the length of the code.

There remain the local ranks.  Write a codeword as
$\epsilon\ww_0+\sum_{j=1}^{U}y_j\ww_j$ with $\epsilon\in\Z_2$ and
$y_j=d_j+2e_j\in\Z_4$, and keep the notation \eqref{eq:symmetric-notation}.
By Lemma \ref{lem:evaluation-principle}, the local rank on a block is the
dimension of the span of the coordinate functions attached to it.

Consider first the block of label $\zero$.  It consists of the $2^{U}$ binary
coordinates of the $\Z_2$ part, and a coordinate there has column
$(1,\lambda_1,\dots,\lambda_U)^T$ with $\lambda\in\Z_2^{U}$ arbitrary.  The
value of the codeword at that position is
$\epsilon+\sum_{j}\lambda_jy_j$ read in $\Z_2$, that is, $\epsilon+\delta_R$
with $R=\{j:\lambda_j=1\}$, since only the bottom digit $d_j$ of $y_j$
survives modulo $2$.  Letting $\lambda$, and hence $R$, vary over all
subsets, these functions span
$
\Span\{\epsilon,d_1,\dots,d_{U}\},
$
a space of dimension $U+1$, its generators being distinct variables.

Consider next a block of nonzero label, and keep the sets $R$ and $R'$ and the
index $j_0$ introduced while counting the blocks. Write $s_j=\lambda'_j+2\mu'_j$ for
the binary expansion of $s_j$ in $\Z_4$, and put $R=\{j:\lambda'_j=1\}$ and
$R'=\{j:\mu'_j=1\}$ as above.  The entry of the codeword at such a
coordinate is $2\epsilon+\sum_jy_js_j$ in $\Z_4$, and
$$
\sum_{j}y_js_j
=\sum_{j}(d_j+2e_j)(\lambda'_j+2\mu'_j)
=\sum_{j}d_j\lambda'_j+2\sum_{j}\bigl(d_j\mu'_j+e_j\lambda'_j\bigr)
\quad\text{in }\Z_4 .
$$
By Lemma \ref{lem:carries}, the ordinary sum $\sum_jd_j\lambda'_j$ has bottom
digit $\delta_R$ and carry $\sym_2(d_R)$, since the bits being added are the
$d_j$ with $j\in R$.  Hence, the bottom and the top digit function of the
coordinate are
$$
\delta_R
\qquad\text{and}\qquad
\psi+\delta_{R'},
\qquad\text{where}\quad
\psi=\epsilon+\eta_R+\sym_2(d_R).
$$
By the normalization recalled above, the set $R'$ runs over all the
subsets of $\{1,\dots,U\}\setminus\{j_0\}$.  Writing
$\delta_{R'}=\sum_{j\not=j_0}\mu'_jd_j$, where the $\mu'_j$ are the free
higher digits of the column, the family of top digits has exactly the shape
\eqref{eq:parameter-polynomial}, with constant term $\psi$ and with the
functions $d_j$, $j\not=j_0$, as the coefficients of the individual
parameters.  Lemma \ref{lem:mobius} therefore shows that the top digits span
$\Span\{\psi,\ d_j\ (j\not=j_0)\}$.  The bottom digit supplies the missing
function: since $j_0\in R$, we have
$\delta_R=d_{j_0}+\sum_{j\in R,\,j\not=j_0}d_j$, and every summand other than
$d_{j_0}$ has just been obtained, so $d_{j_0}$ lies in the span as well.  By Lemma \ref{lem:gray-digits}, 
the local space is therefore $\Span\{d_1,\dots,d_{U},\psi\}$.  Its dimension is
$U+1$: the functions $d_1,\dots,d_U$ are distinct variables, and $\psi$ is the
only one among the $U+1$ listed functions containing the variable $\epsilon$,
so no nontrivial combination of them can vanish.

Hence, for $V=1$, there are $2^{U}$ blocks, all of size $2^{U}$ and all of
local rank $U+1$.

We pass to a general $V$.  Each of the remaining $V-1$ steps of the
construction is the duplication that appends a row of order $2$, whose entries
are $\zero$ on the first copy and equal to $\one$ on the $\Z_2$ part and to
$\two$ on the $\Z_4$ part of the second copy.  Using
$\phi_2(2)=(1,1)$, the Gray image of that row is the vector that vanishes on
the coordinates coming from the first copy and equals $1$ on those coming from
the second one, so the construction is precisely the Plotkin extension
\eqref{eq:plotkin-extension}, up to the permutation of coordinates that lists
the two second copies after the two first ones.  The hypothesis $\one\in\K$
of Lemma \ref{lem:plotkin-blocks} holds for
$H^{U,1}$ because $\Phi(\ww_0)=\one$, and the last assertion of that
lemma supplies it again after each step.  That lemma
therefore applies and, at each step, doubles the number of blocks while
leaving their sizes and their local ranks unchanged.  Applying it $V-1$ times
gives $2^{U}2^{V-1}=2^{U+V-1}$ blocks, all of size $2^{U}$ and all of local
rank $U+1$, which is \eqref{eq:z2z4-profile}.  Note that the number of blocks
is $2^{\kernel(H^{U,V})-1}$ by \eqref{eq:z2z4-invariants}, as one expects,
since the all-one vector lies in the kernel.
\end{proof}

\begin{example}\label{ex:z2z4-profile}
Take $U=2$ and $V=1$.  Then, $t+1=2U+V=5$ by \eqref{eq:z2z4-length}, and
$H^{2,1}$ has length $2^4=16$.  The matrix $A^{2,1}$ has one row of order $2$
and two rows of order $4$; its $\Z_2$ part consists of the four vectors of
$\Z_2^2$ and its $\Z_4$ part consists of the six normalized primitive vectors of
$\Z_4^2$, namely $(1,0)$, $(1,1)$, $(1,2)$, $(1,3)$, $(0,1)$ and $(2,1)$.
There are $2^{U}=4$ labels.  The label $(0,0)$ is carried by the four binary
coordinates of the $\Z_2$ part.  The label $(1,0)$ is carried by the columns
whose parity vector is $(1,0)$, that is, by $(2,1,0)^T$ and $(2,1,2)^T$; the
label $(0,1)$ by $(2,0,1)^T$ and $(2,2,1)^T$; and the label $(1,1)$ by
$(2,1,1)^T$ and $(2,1,3)^T$.  Each of these three labels is thus carried by
two additive coordinates, hence by four binary ones, so all four blocks have
size $2^{U}=4$, as predicted.  On the block of label $\zero$, the local space
is $\Span\{\epsilon,d_1,d_2\}$; on the block of label $(1,1)$, where
$R=\{1,2\}$ and $j_0=1$, it is $\Span\{d_1,d_2,\psi\}$ with
$\psi=\epsilon+e_1+e_2+d_1d_2$.  Both have dimension $U+1=3$, and the same
holds for the two remaining blocks, so
$\cR(H^{2,1})=\{\!\{3^{[4]}\}\!\}$.
\end{example}
\section{The profile of the $\Z_{2^s}$-linear Hadamard codes}
\label{sec:z2s-profile}

For the alphabets $\Z_{2^s}$, no computation like the one of Section
\ref{sec:z2z4-profile} is available, because it would require knowing the rank
of the codes involved, and no rank formula is known for $s\geq4$.  We therefore proceed differently and prove that the profile is constant without
computing any local rank, by producing coordinate permutations that map the
code onto itself and act transitively on its kernel blocks.

We abbreviate $\bar H=\bar H^{a_1,\dots,a_s}$ and
$\bar\cH=\bar\cH^{a_1,\dots,a_s}$ throughout this section, and we keep the
indexing $(\zz,h)$ of the binary coordinates introduced at the end of
Subsection \ref{subsec:z2s-family}.  We produce two families of permutations:
the first one translates the additive coordinate $\zz$ and leaves the Gray
index $h$ untouched, and the second one shifts $h$ and leaves $\zz$ untouched.

\subsection{Two families of permutations preserving the code}
\label{subsec:permutations}

Write $G=\{0\}\times T_{j_2}\times\cdots\times T_{j_m}$, where $j_i$ is the
level of the row $\ww_i$.  This is a subgroup of $\Z_{2^s}^m$, and
\eqref{eq:z2s-columns} identifies the set $\cZ$ of the columns with the coset
$(1,0,\dots,0)+G$.

\begin{lemma}\label{lem:tau}
For every $\delta\in G$, the map $\pi_\delta$ on the binary coordinates
defined by $(\zz,h)\mapsto(\zz+\delta,h)$ is a permutation, and it maps
$\bar H$ onto $\bar H$.
\end{lemma}

\begin{proof}
The map $\zz\mapsto\zz+\delta$ is a bijection of the coset
$\cZ=(1,0,\dots,0)+G$ onto itself, so $\pi_\delta$ is a permutation of the
binary coordinates.

We now check that $\pi_\delta$ maps $\bar H$ onto itself.  Let $\xx$ be a
message and let $\uu\in\bar\cH$ be the corresponding additive codeword, whose
entry at the additive coordinate $\zz$ is $\uu(\zz)=\sum_ix_iz_i$ by
\eqref{eq:z2s-symbol}.  By the definition of the action of a coordinate
permutation, applying $\pi_\delta$ to the binary word $\Phi_s(\uu)$ produces
the word whose value at $(\zz+\delta,h)$ is the value of $\Phi_s(\uu)$ at
$(\zz,h)$; equivalently, its value at $(\zz,h)$ is the value of $\Phi_s(\uu)$
at $(\zz-\delta,h)$, that is, $\phi_s\bigl(\uu(\zz-\delta)\bigr)_h$.  Now
$$
\uu(\zz-\delta)=\sum_{i=1}^{m}x_i(z_i-\delta_i)
=\sum_{i=1}^{m}x_iz_i-\sum_{i=1}^{m}x_i\delta_i
=\uu(\zz)-\lambda,
$$
where $\lambda=\sum_{i=2}^{m}x_i\delta_i\in\Z_{2^s}$, the term $i=1$ being absent 
because $\delta_1=0$ by the definition of $G$. The crucial point is that $\lambda$ 
depends only on the message $\xx$ and on $\delta$, and not on $\zz$; it is therefore 
the same constant at every additive coordinate.  Hence, the image word is $\Phi_s(\uu-\lambda\one)$. 
Since $\ww_1=\one$ is a row of $\bar A^{a_1,\dots,a_s}$, the constant word $\lambda\one=\lambda\ww_1$ 
belongs to $\bar\cH$, and therefore $\uu-\lambda\one\in\bar\cH$. So, the image word belongs to
$\bar H=\Phi_s(\bar\cH)$.  As $\pi_\delta$ is a bijection, it maps $\bar H$
onto $\bar H$.
\end{proof}

The second family shifts the Gray coordinates inside every additive
coordinate.  Here is the point where the parameter $\sigma$ of
\eqref{eq:sigma} enters: not every shift preserves the code, and $\sigma$ of
them are shown below to do so, which is all that the proof needs.  For
$\uu\in\bar\cH$ and $0\leq\nu\leq s-1$, we write $\uu^{(\nu)}$ for the binary
vector formed by the $\nu$th binary digits of the entries of $\uu$.

We will repeatedly use the identity
\begin{equation}\label{eq:parity-identity}
2^{s-1}(w\bmod 2)=2^{s-1}w\quad\text{in }\Z_{2^s}, 
\end{equation}
for every $w$, which holds because the difference of the two sides is $2^{s-1}$ times an even
integer, and $2^{s-1}\cdot2=2^{s}=0$ in $\Z_{2^s}$.

\begin{lemma}\label{lem:theta}
Let $\bar H=\bar H^{a_1,\dots,a_s}$ be nonlinear.  For $0\leq\nu\leq s-2$,
let $\theta_\nu$ be the permutation of the binary coordinates defined by
$(\zz,h)\mapsto(\zz,h+\bfe_\nu)$, where $\bfe_\nu$ is the $\nu$th vector of
the standard basis of $\Z_2^{s-1}$.  Then, for every $\uu\in\bar\cH$,
\begin{equation}\label{eq:theta-action}
\theta_\nu\bigl(\Phi_s(\uu)\bigr)=\Phi_s\bigl(\uu+2^{s-1}\uu^{(\nu)}\bigr),
\end{equation}
and $\theta_\nu$ maps $\bar H$ onto $\bar H$ for every $\nu$ with
$0\leq\nu\leq\sigma-1$.
\end{lemma}

\begin{proof}
We prove \eqref{eq:theta-action} first.  By the definition of the action of a
coordinate permutation, applying $\theta_\nu$ to $\Phi_s(\uu)$ produces the
word whose value at $(\zz,h+\bfe_\nu)$ is the value of $\Phi_s(\uu)$ at
$(\zz,h)$.  Since $\bfe_\nu+\bfe_\nu=\zero$ in $\Z_2^{s-1}$, the map
$\theta_\nu$ is an involution, so the image word takes at $(\zz,h)$ the value
of $\Phi_s(\uu)$ at $(\zz,h+\bfe_\nu)$, that is,
$\phi_s\bigl(\uu(\zz)\bigr)_{h+\bfe_\nu}$.  In the expression
\eqref{eq:z2s-coordinate-value} for $\phi_s(w)_h$, the index $h$ appears only
through the sum $\sum_{i\leq s-2}h_iw_i$, and replacing $h$ by $h+\bfe_\nu$
changes exactly the $\nu$th summand, adding $w_\nu$ to it.  Hence,
$$
\phi_s\bigl(\uu(\zz)\bigr)_{h+\bfe_\nu}
=\phi_s\bigl(\uu(\zz)\bigr)_h+\uu(\zz)_\nu ,
$$
where $\uu(\zz)_\nu$ is the $\nu$th binary digit of the entry of $\uu$ at
$\zz$.  On the other hand, $\phi_s(2^{s-1}\varepsilon)$ is the all-one vector
when $\varepsilon=1$ and the all-zero vector when $\varepsilon=0$, directly
from \eqref{eq:Carlet}; so the word $\Phi_s(2^{s-1}\uu^{(\nu)})$ takes at
$(\zz,h)$ the value $\uu(\zz)_\nu$, whatever $h$ is.  Therefore,
$\theta_\nu(\Phi_s(\uu))=\Phi_s(\uu)+\Phi_s(2^{s-1}\uu^{(\nu)})$, and by Lemma
\ref{lem:gray-shift}, applied at each additive coordinate separately, the
right-hand side equals $\Phi_s\bigl(\uu+2^{s-1}\uu^{(\nu)}\bigr)$.  This is
\eqref{eq:theta-action}.

Consequently, $\theta_\nu$ maps $\bar H$ into itself, hence onto itself, as
soon as $\uu+2^{s-1}\uu^{(\nu)}\in\bar\cH$ for every $\uu\in\bar\cH$.  We now
verify this condition for $0\leq\nu\leq\sigma-1$, distinguishing the case
$\nu=0$ from the case $\nu\geq1$.

Let $\nu=0$.  Since $\uu^{(0)}$ is the
vector of the bottom digits of $\uu$, that is, of the entries of $\uu$ reduced
modulo $2$, applying \eqref{eq:parity-identity} entrywise gives
$2^{s-1}\uu^{(0)}=2^{s-1}\uu$.  Hence,
$$
\uu+2^{s-1}\uu^{(0)}=\uu+2^{s-1}\uu=(1+2^{s-1})\uu\in\bar\cH ,
$$
since $\bar\cH$ is a $\Z_{2^s}$-module.  This settles the case
$\sigma=1$, where $\nu=0$ is the only admissible index.

Let now $\sigma\geq2$ and $1\leq\nu\leq\sigma-1$.  By
\eqref{eq:sigma}, $\sigma\geq2$ forces $a_1=1$. Moreover, the second case of
\eqref{eq:sigma} is excluded, because it describes a linear code, so
$\sigma=\min\{i:a_i>0,\ 2\leq i\leq s\}$.  Hence, $\ww_1=\one$ is the only
row of level $1$, every other row has level at least $\sigma$, and by
\eqref{eq:Ti} the entries of those rows lie in
$T_\sigma=2^{\sigma-1}\Z_{2^s}$.  Write accordingly
$\uu=x\one+\uu'$ with $x\in\Z_{2^s},$
where $\uu'$ is the part of $\uu$ generated by $\ww_2,\dots,\ww_m$; hence, every
entry of $\uu'$ therefore lies in $2^{\sigma-1}\Z_{2^s}$.  Since
$\nu\leq\sigma-1$, we may factor $\uu'=2^{\nu}\uu''$ for some word $\uu''$
over $\Z_{2^s}$.

We compute the digit of index $\nu$ of an entry of $\uu$.  At the additive
coordinate $\zz$ that entry is $x+2^{\nu}\uu''(\zz)$.  Adding a multiple of
$2^\nu$ does not affect the digits of index smaller than $\nu$, and the digit
of index $\nu$ of $x+2^{\nu}\uu''(\zz)$ is the bottom digit of
$\lfloor x/2^{\nu}\rfloor+\uu''(\zz)$, which is $x_\nu+\uu''(\zz)_0$ in
$\Z_2$.  In vector form, $\uu^{(\nu)}=x_\nu\one+(\uu''\bmod2)$. Therefore, using
\eqref{eq:parity-identity} and then $\uu'=2^\nu\uu''$, we obtain
\begin{equation*}
\begin{aligned}
2^{s-1}\uu^{(\nu)}
&=2^{s-1}x_\nu\one+2^{s-1}\bigl(\uu''\bmod2\bigr)
 =2^{s-1}x_\nu\one+2^{s-1}\uu''\\
&=2^{s-1}x_\nu\one+2^{\,s-1-\nu}\,2^{\nu}\uu''
 =2^{s-1}x_\nu\one+2^{\,s-1-\nu}\uu' .
\end{aligned}
\end{equation*}
Adding $\uu=x\one+\uu'$ to both sides and grouping the two parts,
$
\uu+2^{s-1}\uu^{(\nu)}=\bigl(x+2^{s-1}x_\nu\bigr)\one
 +\bigl(1+2^{\,s-1-\nu}\bigr)\uu' .
$
By Lemma \ref{lem:sigma-bound}, we have $\nu\leq\sigma-1\leq s-2$, so
$s-1-\nu\geq1$ and $2^{\,s-1-\nu}$ is even; hence $1+2^{\,s-1-\nu}$ is odd and
hence a unit of $\Z_{2^s}$.  Consequently,
$(1+2^{\,s-1-\nu})\uu'$ lies in the subcode generated by $\ww_2,\dots,\ww_m$,
and $(x+2^{s-1}x_\nu)\one$ lies in the subcode generated by $\ww_1$; hence, their sum
lies in $\bar\cH$, which is what we had to prove.
\end{proof}

\subsection{The kernel blocks, and transitivity}\label{subsec:z2s-blocks}

The third ingredient describes the kernel blocks.  We determine them directly from the basis
\eqref{eq:z2s-kernel-basis} of $\K(\bar H)$ recalled in Proposition
\ref{prop:z2s-kernel-basis}.  We first fix some notation.  Let $j_i$ be
the level of the row $\ww_i$.  For $2\leq i\leq m$, the entry $z_i$ of an
additive coordinate $\zz$ lies in $T_{j_i}=2^{\,j_i-1}\Z_{2^s}$ by
\eqref{eq:z2s-columns}, so we may write $z_i=2^{\,j_i-1}\zeta_i$ with
$\zeta_i$ determined modulo $2^{\,s-j_i+1}$.  Since $s-j_i+1\geq1$, the parity
$\chi_i(\zz)=\zeta_i\bmod2$ does not depend on the representative $\zeta_i$,
and we write $\chi(\zz)=\bigl(\chi_2(\zz),\dots,\chi_m(\zz)\bigr)\in\Z_2^{m-1}$.

\begin{lemma}\label{lem:z2s-kernel-structure}
Let $\bar H$ be nonlinear.  Then, two binary coordinates $(\zz,h)$
and $(\zz',h')$ lie in the same kernel block if and only if
\begin{equation}\label{eq:z2s-block-condition}
\chi_i(\zz)=\chi_i(\zz')\ (2\leq i\leq m)
\qquad\text{and}\qquad
h+h'\in L_\sigma^{\perp},
\end{equation}
where
\begin{equation}\label{eq:Lsigma}
L_\sigma=\Span\bigl\{\bfe_0,\dots,\bfe_{\sigma-2},\
 \bfe_{\sigma-1}+\bfe_{\sigma}+\cdots+\bfe_{s-2}\bigr\}\subseteq\Z_2^{s-1},
\qquad \dim L_\sigma=\sigma .
\end{equation}
Moreover, the map $\zz\mapsto\chi(\zz)$ is onto $\Z_2^{m-1}$, and the number
of kernel blocks is $2^{\,\kernel(\bar H)-1}$.
\end{lemma}

\begin{proof}
By \eqref{eq:kernel-relation} and the convention recorded after Definition
\ref{def:kernel-block-profile}, two binary coordinates lie in the same kernel
block precisely when every vector of the basis \eqref{eq:z2s-kernel-basis}
takes the same value on them.  We therefore evaluate that basis at a binary
coordinate $(\zz,h)$.

Consider first the $m$ basis vectors
$\Phi_s(\frac{\ord(\ww_i)}{2}\ww_i)$, and let $\ww_i$ have level $j_i$.  Then,
$\ord(\ww_i)=2^{s-j_i+1}$, and hence
$\frac{\ord(\ww_i)}{2}\ww_i=2^{\,s-j_i}\ww_i$.  Its entry at the additive
coordinate $\zz$ is $2^{\,s-j_i}z_i$, and, writing $z_i=2^{\,j_i-1}\zeta_i$ as above, this entry equals
$$
2^{\,s-j_i}\cdot2^{\,j_i-1}\zeta_i=2^{\,s-1}\zeta_i=2^{\,s-1}\chi_i(\zz),
$$
the last equality holding by \eqref{eq:parity-identity}.  Its Gray image is
therefore the all-one vector of length $2^{s-1}$ when $\chi_i(\zz)=1$ and the
all-zero vector when $\chi_i(\zz)=0$; in other words, the $i$th basis vector
takes at $(\zz,h)$ the value $\chi_i(\zz)$, which does not depend on $h$.  For
$i=1$, we have $j_1=1$ and $z_1=1$, so $\chi_1(\zz)=1$ for every $\zz$ and the
first basis vector is $\one$, which separates no two coordinates.

Consider next the last $\sigma$ basis vectors.  Each of them has the form
$\Phi_s(\lambda\one)$ for one of the $\sigma$ elements $\lambda$ listed in
\eqref{eq:z2s-kernel-basis}, namely $\lambda=2^{\nu}$ with
$0\leq\nu\leq\sigma-2$ and $\lambda=2^{s-1}-1$.  Since $\lambda\one$ has the
same entry $\lambda$ at every additive coordinate,
\eqref{eq:z2s-coordinate-value} gives
$$
\bigl(\Phi_s(\lambda\one)\bigr)_{(\zz,h)}
=\lambda_{s-1}+\sum_{i\leq s-2}\lambda_ih_i ,
$$
which does not depend on $\zz$.

Putting the two computations together, the
coordinates $(\zz,h)$ and $(\zz',h')$ lie in the same block exactly when
$\chi_i(\zz)=\chi_i(\zz')$ for $2\leq i\leq m$, which is the first condition
of \eqref{eq:z2s-block-condition}, and when
$$
\lambda_{s-1}+\sum_{i\leq s-2}\lambda_ih_i
=\lambda_{s-1}+\sum_{i\leq s-2}\lambda_ih'_i,
\quad\text{that is,}\quad
\sum_{i\leq s-2}\lambda_i(h_i+h'_i)=0,
$$
for each of the $\sigma$ elements $\lambda$.  The left-hand side is the
standard inner product of $h+h'$ with the projection of the digit vector of
$\lambda$ onto the first $s-1$ coordinates, so the second condition says that
$h+h'$ is orthogonal to all those projections, that is, to their span.  Let us
identify that span.  For $\lambda=2^{\nu}$ with $\nu\leq\sigma-2$, the digit
vector is $\bfe_\nu$, and its projection is $\bfe_\nu$ itself, since
$\nu\leq\sigma-2\leq s-2$.  For $\lambda=2^{s-1}-1$, the binary digits equal
$1$ in the positions $0,\dots,s-2$ and $0$ in the position $s-1$, so the
projection is $\bfe_0+\cdots+\bfe_{s-2}$.  Subtracting from that last vector
the vectors $\bfe_0,\dots,\bfe_{\sigma-2}$, which are already in the span,
leaves $\bfe_{\sigma-1}+\cdots+\bfe_{s-2}$.  Hence, the span is exactly the
space $L_\sigma$ displayed in \eqref{eq:Lsigma}, and the second condition
reads $h+h'\in L_\sigma^{\perp}$.

Moreover, $\dim L_\sigma=\sigma$: the remaining
generator $\bfe_{\sigma-1}+\cdots+\bfe_{s-2}$ is supported in
$\{\sigma-1,\dots,s-2\}$, which is disjoint from $\{0,\dots,\sigma-2\}$
and nonempty because $\sigma\leq s-1$ by Lemma \ref{lem:sigma-bound}.  That
generator is therefore nonzero and independent of the previous ones, so
$\dim L_\sigma=(\sigma-1)+1=\sigma$.

It remains to see that every admissible label occurs, and to count the blocks.  By \eqref{eq:z2s-columns}, the
entries $z_i$, for $2\leq i\leq m$, range independently over $T_{j_i}$; and,
as $z_i=2^{\,j_i-1}\zeta_i$ runs over $T_{j_i}$, the element $\zeta_i$ runs
over $\Z_{2^{\,s-j_i+1}}$, which contains elements of both parities because
$s-j_i+1\geq1$.  Hence, $\chi_i(\zz)$ takes both values, independently for the
different $i$, and the map $\zz\mapsto(\chi_2(\zz),\dots,\chi_m(\zz))$ is onto
$\Z_2^{m-1}$.  Finally, by \eqref{eq:z2s-block-condition}, a block is
determined by the value of $(\chi_2(\zz),\dots,\chi_m(\zz))\in\Z_2^{m-1}$
together with the class of $h$ modulo $L_\sigma^{\perp}$, and the latter takes
$|\Z_2^{s-1}/L_\sigma^{\perp}|=|L_\sigma|=2^{\sigma}$ values. Since all these
combinations occur, the number of blocks is
$2^{m-1}2^{\sigma}=2^{\,\sigma+m-1}$, which is $2^{\,\kernel(\bar H)-1}$ by
\eqref{eq:z2s-kernel}.
\end{proof}

\begin{proposition}\label{prop:z2s-transitive}
Let $s\geq2$ and let $\bar H=\bar H^{a_1,\dots,a_s}$ be a nonlinear
$\Z_{2^s}$-linear Hadamard code, and let $\Gamma$ be the group of permutations
of the binary coordinates generated by the maps $\pi_\delta$ with $\delta\in
G$ and by the maps $\theta_\nu$ with $0\leq\nu\leq\sigma-1$.  Then, every
element of $\Gamma$ maps $\bar H$ onto $\bar H$, and $\Gamma$ acts
transitively on the set of kernel blocks of $\bar H$.  Consequently, for any
two kernel blocks $B$ and $B'$ of $\bar H$, the punctured codes
$\bar H|_{B}$ and $\bar H|_{B'}$ are permutation equivalent; in particular
they have the same length and the same local rank.
\end{proposition}

\begin{proof}
By Lemmas \ref{lem:tau} and \ref{lem:theta}, every
generator of $\Gamma$ maps $\bar H$ onto $\bar H$, hence so does every element
of $\Gamma$.  Applying Lemma \ref{lem:profile-invariant} with $C=D=\bar H$ and
with the translation vector $\yy=\zero$, we conclude that every element of
$\Gamma$ permutes the kernel blocks of $\bar H$.  It is therefore enough to
prove that $\Gamma$ acts transitively on the set of blocks.

We first attach a convenient label to the blocks.  Taking the inner products
of a vector $h=(h_0,\dots,h_{s-2})$ of $\Z_2^{s-1}$ with the basis vectors of
$L_\sigma$ displayed in \eqref{eq:Lsigma} defines the map
\begin{equation}\label{eq:sel}
\sel:\Z_2^{s-1}\rightarrow\Z_2^{\sigma},
\qquad
\sel(h)=\bigl(h_0,\dots,h_{\sigma-2},\
 h_{\sigma-1}+\cdots+h_{s-2}\bigr).
\end{equation}
Thus, $\sel(h)=\zero$ if and only if $h$ is orthogonal to every basis vector
of $L_\sigma$, and hence
$$
\{h\in\Z_2^{s-1}:\ \sel(h)=\zero\}=L_\sigma^{\perp} .
$$
Consequently, $h+h'$ lies in $L_\sigma^{\perp}$ if and only if
$\sel(h)=\sel(h')$.  Condition \eqref{eq:z2s-block-condition} of Lemma
\ref{lem:z2s-kernel-structure} then says that the block of a binary
coordinate $(\zz,h)$ is determined by the pair $\bigl(\chi(\zz),\sel(h)\bigr)$,
and, as shown at the end of that lemma, every such pair occurs.

We deal with the first component.  The maps $\pi_\delta$ fix $h$, so they do not
change $\sel(h)$, and they replace $\zz$ by $\zz+\delta$.  Since
$\chi_i$ is additive, being the composition of the division by $2^{\,j_i-1}$
with reduction modulo $2$, this replaces $\chi(\zz)$ by
$\chi(\zz)+\chi(\delta)$.  As $\delta_i$ runs over
$T_{j_i}=2^{\,j_i-1}\Z_{2^s}$, the bit $\chi_i(\delta)$ takes both values, and
it does so independently for the different $i$, exactly as in the last part of
the proof of Lemma \ref{lem:z2s-kernel-structure}; hence
$\{\chi(\delta):\delta\in G\}=\Z_2^{m-1}$.  The maps $\pi_\delta$ therefore
act transitively on the first component of the block label, while fixing the
second one.

We deal now with the second component.  The maps $\theta_\nu$ fix $\zz$, so
they do not
change $\chi(\zz)$, and they replace $h$ by $h+\bfe_\nu$, hence $\sel(h)$ by
$\sel(h)+\sel(\bfe_\nu)$ because $\sel$ is linear.  By Lemma
\ref{lem:sigma-bound}, we have $\sigma-1\leq s-2$, so the vectors
$\bfe_0,\dots,\bfe_{\sigma-1}$ all lie in $\Z_2^{s-1}$ and all the
corresponding maps $\theta_\nu$ are available.  Evaluating \eqref{eq:sel} at
$h=\bfe_\nu$ shows that $\sel(\bfe_\nu)$ is the $\nu$th standard basis vector
of $\Z_2^{\sigma}$ for every $\nu$ with $0\leq\nu\leq\sigma-1$.  Indeed, for
$\nu\leq\sigma-2$ the last coordinate of $\sel(\bfe_\nu)$ vanishes and the
$\nu$th one equals $1$; and for $\nu=\sigma-1$ the
first $\sigma-1$ coordinates vanish and the last one equals $1$.  The $\sigma$
vectors $\sel(\bfe_\nu)$ are therefore the standard basis of $\Z_2^{\sigma}$
and they span it, so the maps $\theta_\nu$ act transitively on the second
component of the block label, while fixing the first one.

Given two blocks, we may therefore first apply a suitable $\pi_\delta$ to
match the first components of their labels, and then a suitable composition of
maps $\theta_\nu$ to match the second components, which does not disturb the
first ones.  Hence, $\Gamma$ acts transitively on the set of blocks.

It remains to compare the punctured codes.  Let $B$ and $B'$ be two kernel
blocks and let $\pi\in\Gamma$ be such that $\pi(B)=B'$, and let
$\pi_B:B\rightarrow B'$ be the restriction of $\pi$, which is a bijection.
As in the observation preceding \eqref{eq:relabel}, the vector
$\pi(\xx)|_{B'}$ is obtained from $\xx|_{B}$ by relabelling the positions
through $\pi_B$.  Since $\pi(\bar H)=\bar H$, the code $\bar H|_{B'}$ is therefore obtained from
$\bar H|_{B}$ by that relabelling, so the two are permutation equivalent.  In
particular $|B|=|B'|$, and the two punctured codes have the same dimension of
the linear span, which by \eqref{eq:span-puncture} is the local rank. \qed
\end{proof}

\begin{theorem}\label{thm:z2s-profile-constant}
Let $s\geq2$ and let $\bar H=\bar H^{a_1,\dots,a_s}$ be a nonlinear
$\Z_{2^s}$-linear Hadamard code of length $2^t$.  Then, all the kernel blocks
of $\bar H$ have the same size and the same local rank.  Equivalently, there
is an integer $\rho_{a_1,\dots,a_s}$ such that
\begin{equation}\label{eq:z2s-profile}
\cR\bigl(\bar H^{a_1,\dots,a_s}\bigr)
=\bigl\{\!\bigl\{\rho_{a_1,\dots,a_s}^{\,
 [2^{\,\kernel(\bar H)-1}]}\bigr\}\!\bigr\},
\end{equation}
that is, the kernel-block rank profile of $\bar H$ is constant.
\end{theorem}

\begin{proof}
By Proposition \ref{prop:z2s-transitive}, all the kernel blocks of $\bar H$
have the same size and the same local rank, say $\rho_{a_1,\dots,a_s}$.
Since the number of blocks is $2^{\,\kernel(\bar H)-1}$ by Lemma
\ref{lem:z2s-kernel-structure}, we obtain \eqref{eq:z2s-profile}. \qed
\end{proof}

\begin{example}\label{ex:z2s-blocks}
Let $s=3$ and $(a_1,a_2,a_3)=(2,0,0)$, which makes $\bar H^{2,0,0}$ a
$\Z_8$-linear Hadamard code of length $2^5=32$; here $m=2$, both rows have
level $1$, and $\sigma=1$ because $a_1\geq2$.  The matrix
$\bar A^{2,0,0}$ has the eight columns $(1,j)^T$ with $j\in\Z_8$.  Since
$\sigma=1$, formula \eqref{eq:Lsigma} gives
$L_1=\Span\{\bfe_0+\bfe_1\}$, so $\sel(h)=h_0+h_1$ by \eqref{eq:sel} and, by
Lemma \ref{lem:z2s-kernel-structure}, the block of the binary coordinate
$(j,h)$ is determined by the pair $\bigl(j\bmod2,\ h_0+h_1\bigr)$.  There are
thus $4=2^{\kernel-1}$ blocks, in accordance with $\kernel(\bar H^{2,0,0})=3$
given by \eqref{eq:z2s-kernel}, and each of them contains
$32/4=8$ binary coordinates.  The permutations of Lemma \ref{lem:tau} are here
$(j,h)\mapsto(j+\delta,h)$ with $\delta\in\Z_8$, and taking $\delta$ odd
changes the first component of the label. The single permutation
$\theta_0:(j,h)\mapsto(j,h+\bfe_0)$ of Lemma \ref{lem:theta} changes the
second one.  Together they move any block onto any other, so the profile is
constant, as Theorem \ref{thm:z2s-profile-constant} asserts.
\end{example}

\section{The punctured code on a kernel block}\label{sec:descent}

Theorem \ref{thm:z2s-profile-constant} says that all the kernel blocks of
$\bar H^{a_1,\dots,a_s}$ carry the same local rank, but not what that common
value is.  We now identify it, in a form much stronger than a numerical
formula: the punctured code on a kernel block is the $\Z_{2^{s-1}}$-linear
Hadamard code $\bar H^{a_1,\dots,a_{s-1}}$, one alphabet lower.  This is what
makes the case $t_1=1$ of Theorem \ref{thm:cross-family} accessible for every
$s$.

Throughout the section, we assume that $\bar H^{a_1,\dots,a_s}$ is nonlinear
and that $a_1\geq2$, and then $\sigma=1$ by \eqref{eq:sigma}.  For $s\geq3$,
the second assumption implies the first, by the linearity criterion recalled
in Subsection \ref{subsec:z2s-family}, and for $s=2$, the two together mean
$a_1\geq3$.  The hypothesis $a_1\geq2$ costs nothing, as Remark
\ref{rem:chain-reduction} explains at the end of the section.  Since all the
blocks are equivalent to
one another by Theorem \ref{thm:z2s-profile-constant}, it is enough to work
with one of them, and we choose the most convenient one.

\begin{lemma}\label{lem:zero-block}
Let $\bar H^{a_1,\dots,a_s}$ be nonlinear with $a_1\geq2$.  Then, the kernel
block of $\bar H^{a_1,\dots,a_s}$ containing the binary coordinate
$\bigl((1,0,\dots,0),\zero\bigr)$ is
\begin{equation}\label{eq:zero-block}
B_0=\bigl\{(\zz,h):\ z_i\in 2^{j_i}\Z_{2^s}\ \text{ for } 2\leq i\leq m,
\ \text{ and }\ \wt_H(h)\ \text{is even}\bigr\}.
\end{equation}
\end{lemma}

\begin{proof}
Since $a_1\geq2$, we have $\sigma=1$, so the middle family in
\eqref{eq:Lsigma} is empty and
$L_1=\Span\{\bfe_0+\bfe_1+\cdots+\bfe_{s-2}\}$.  Write
$\zz_\ast=(1,0,\dots,0)$ for the additive coordinate in the statement; its
entries satisfy $z_{\ast,i}=0$ for $i\geq2$, so $\chi_i(\zz_\ast)=0$ for
$2\leq i\leq m$.

By \eqref{eq:z2s-block-condition}, a binary coordinate $(\zz,h)$ lies in the
same block as $(\zz_\ast,\zero)$ precisely when $\chi_i(\zz)=\chi_i(\zz_\ast)
=0$ for $2\leq i\leq m$, and when $h+\zero=h$ belongs to $L_1^{\perp}$.  We
translate the two conditions.  The first one says that
$z_i/2^{\,j_i-1}$ is even, that is, that $z_i$ is a multiple of $2^{\,j_i}$;
since $z_i$ already lies in $T_{j_i}=2^{\,j_i-1}\Z_{2^s}$, this is exactly
$z_i\in2^{\,j_i}\Z_{2^s}$.  The second one says that the inner product of $h$
with the single generator $\bfe_0+\cdots+\bfe_{s-2}$ of $L_1$ vanishes, that
is, $h_0+h_1+\cdots+h_{s-2}=0$, which means that $\wt_H(h)$ is even.  This is
\eqref{eq:zero-block}.
\end{proof}

Deleting from $\bar A^{a_1,\dots,a_s}$ the $a_s$ rows of level $s$ leaves the
matrix $\bar A^{a_1,\dots,a_{s-1}}$, which therefore has $m'=m-a_s$ rows.  We
now show that puncturing on $B_0$ realizes exactly that deletion.

\begin{theorem}\label{thm:descent}
Let $s\geq2$ and let $\bar H^{a_1,\dots,a_s}$ be a nonlinear
$\Z_{2^s}$-linear Hadamard code with $a_1\geq2$.  Put $m'=m-a_s$ and write every binary coordinate of the block $B_0$ of \eqref{eq:zero-block} as $(\zz,h)$ with $z_i=2^{j_i}\zeta_i$, where
$\zeta_i\in\Z_{2^{s-j_i}}$ for $2\leq i\leq m$, and with $\wt_H(h)$ even,
and put $h'=(h_1,\dots,h_{s-2})\in\Z_2^{s-2}$, that is, $h'_\mu=h_{\mu+1}$
for $0\leq\mu\leq s-3$.  Then, the map
\begin{equation}\label{eq:descent-bijection}
\beta:\ (\zz,h)\longmapsto(\zz',h'),
\qquad
\zz'=\bigl(1,\,2^{\,j_2-1}\zeta_2,\,\dots,\,2^{\,j_{m'}-1}\zeta_{m'}\bigr)
\in\Z_{2^{s-1}}^{\,m'},
\end{equation}
is a bijection from $B_0$ onto the set of binary coordinates of the
$\Z_{2^{s-1}}$-linear Hadamard code $\bar H^{a_1,\dots,a_{s-1}}$, and it
carries the punctured code onto that code:
\begin{equation}\label{eq:descent}
\beta\Bigl(\bar H^{a_1,\dots,a_s}\big|_{B_0}\Bigr)
=\bar H^{a_1,\dots,a_{s-1}} .
\end{equation}
Here, for $s=2$, the code $\bar H^{a_1}$ over $\Z_2$ is understood to be the
binary linear Hadamard code of length $2^{\,a_1-1}$.
\end{theorem}

\begin{proof}
The proof has three parts: first we check that $\beta$ is a bijection onto the
announced set, then we compute how the values of the codewords transform, and
finally we identify the resulting set of binary words.

We begin with the bijectivity of $\beta$.  Let
$T'_j=2^{\,j-1}\Z_{2^{s-1}}$ denote the sets \eqref{eq:Ti} for the alphabet
$\Z_{2^{s-1}}$. By \eqref{eq:z2s-columns} applied to
$\bar A^{a_1,\dots,a_{s-1}}$, the columns of that matrix are exactly the
vectors $(1,z'_2,\dots,z'_{m'})$ whose entry in a row of level $j$ lies in
$T'_j$.  We compare the two sides entry by entry.

Fix $i$ with $2\leq i\leq m$ and suppose first that $j_i\leq s-1$, which means
that the $i$th row survives in $\bar A^{a_1,\dots,a_{s-1}}$.  The condition
$z_i\in2^{\,j_i}\Z_{2^s}$ of \eqref{eq:zero-block} says that
$z_i=2^{\,j_i}\zeta_i$; here $\zeta_i$ is determined modulo $2^{\,s-j_i}$,
because $2^{\,j_i}\zeta_i$ is computed modulo $2^{s}$, so $\zeta_i$ runs
exactly over $\Z_{2^{\,s-j_i}}$, a set of $2^{\,s-j_i}$ elements.  On the
other side, the entry $z'_i=2^{\,j_i-1}\zeta_i$ runs over
$T'_{j_i}=2^{\,j_i-1}\Z_{2^{s-1}}$, which has
$2^{\,(s-1)-(j_i-1)}=2^{\,s-j_i}$ elements, and it does so bijectively,
because $2^{\,j_i-1}\zeta_i$ modulo $2^{s-1}$ determines and is determined by
$\zeta_i$ modulo $2^{\,s-j_i}$.  Suppose now that $j_i=s$, which happens for
the last $a_s$ rows.  Then, $T_s=\{0,2^{s-1}\}$ and the condition
$z_i\in2^{s}\Z_{2^s}=\{0\}$ forces $z_i=0$, hence $\zeta_i=0$; those rows are
therefore constant on $B_0$, in agreement with the
fact that $\bar A^{a_1,\dots,a_{s-1}}$ has no row of level $s$.  This is why
they are simply omitted in \eqref{eq:descent-bijection}.

It remains to treat the Gray index.  As $\wt_H(h)$ is even, the entry
$h_0=h_1+\cdots+h_{s-2}$ is determined by the other ones, so $h\mapsto h'$ is
a bijection from the even-weight vectors of $\Z_2^{s-1}$ onto $\Z_2^{s-2}$. Combining the two
observations, $\beta$ is a bijection from $B_0$ onto the set of pairs
$(\zz',h')$ with $\zz'$ a column of $\bar A^{a_1,\dots,a_{s-1}}$ and
$h'\in\Z_2^{s-2}$, that is, onto the set of binary coordinates of
$\bar H^{a_1,\dots,a_{s-1}}$.

We turn to the values of the codewords.  Let $\xx=(x_1,\dots,x_m)$ be a message and let
$\uu\in\bar\cH^{a_1,\dots,a_s}$ be the corresponding additive codeword.  At an
additive coordinate $\zz$ of $B_0$, we may substitute $z_1=1$ and
$z_i=2^{\,j_i}\zeta_i$ into \eqref{eq:z2s-symbol}, which gives
$$
W:=\uu(\zz)=x_1+\sum_{i=2}^{m}x_iz_i
        =x_1+\sum_{i=2}^{m}2^{\,j_i}x_i\zeta_i
        =x_1+2\sum_{i=2}^{m}2^{\,j_i-1}x_i\zeta_i \pmod{2^s},
$$
where the exponents $j_i-1$ are nonnegative because every level is at least
$1$.  Since $2^{\,j_i}x_i\zeta_i$ modulo $2^s$ depends on $x_i$ only through
$x_i$ modulo $2^{\,s-j_i}$, we may replace $x_i$ by
$u_i=x_i\bmod 2^{\,s-j_i}$ without changing $W$.  Writing also
$x_1=\varepsilon+2q$
with $\varepsilon\in\{0,1\}$ and $q\in\Z_{2^{s-1}}$, we obtain
\begin{equation}\label{eq:W-and-Wprime}
W=\varepsilon+2W',\qquad
W':=q+\sum_{i=2}^{m}2^{\,j_i-1}u_i\zeta_i\ \in\Z_{2^{s-1}} .
\end{equation}
Reading off the binary digits of both sides of $W=\varepsilon+2W'$ gives
$W_0=\varepsilon$ and $W_\nu=W'_{\nu-1}$ for $1\leq\nu\leq s-1$.  We can now
compute the value of the codeword at the binary coordinate $(\zz,h)$.  Using
\eqref{eq:z2s-coordinate-value}, then the substitution just made, and finally
$h_0=h_1+\cdots+h_{s-2}=\sum_{\mu=0}^{s-3}h'_\mu$, we get
\begin{equation}\label{eq:descent-value}
\begin{aligned}
\phi_s(W)_h
&=W_{s-1}+h_0W_0+\sum_{\nu=1}^{s-2}h_\nu W_\nu\\
&=W'_{s-2}+\varepsilon\,h_0+\sum_{\nu=1}^{s-2}h_\nu W'_{\nu-1}\\
&=W'_{s-2}+\sum_{\mu=0}^{s-3}h'_\mu W'_\mu
   +\varepsilon\sum_{\mu=0}^{s-3}h'_\mu
 =\phi_{s-1}(W')_{h'}+\varepsilon\,\wt_H(h'),
\end{aligned}
\end{equation}
where in the third line we set $\mu=\nu-1$, and where the last equality uses
\eqref{eq:z2s-coordinate-value} written for the alphabet $\Z_{2^{s-1}}$.

It remains to identify the punctured code.  We claim that $W'$ is exactly
the value, at the additive coordinate $\zz'$ of \eqref{eq:descent-bijection},
of the codeword of $\bar\cH^{a_1,\dots,a_{s-1}}$ determined by the message
$(q,u_2,\dots,u_{m'})$.  Indeed, the first row of
$\bar A^{a_1,\dots,a_{s-1}}$ is the all-one row and contributes $q$; the $i$th
row contributes $u_iz'_i=2^{\,j_i-1}u_i\zeta_i$; and the rows of level $s$,
which have been dropped, contributed nothing to $W'$ because their $\zeta_i$
vanish.  Comparing with \eqref{eq:W-and-Wprime} proves the claim.  Moreover, a
coefficient of a row of level $j$ over $\Z_{2^{s-1}}$ ranges over
$\Z_{2^{s-j}}$, because that row has order $2^{\,(s-1)-j+1}=2^{\,s-j}$, and
this is exactly the range of $u_i$.

Finally, as $\xx$ runs over all the messages of $\bar\cH^{a_1,\dots,a_s}$, the
pair $\bigl(\varepsilon,(q,u_2,\dots,u_{m'})\bigr)$ runs over $\Z_2\times\cM'$,
where $\cM'$ is the message set of $\bar\cH^{a_1,\dots,a_{s-1}}$, and its two
components vary independently: the map $x_1\mapsto(\varepsilon,q)$ is a
bijection
from $\Z_{2^s}$ onto $\Z_2\times\Z_{2^{s-1}}$, and $x_i\mapsto u_i$ is
surjective for every $i\geq2$.  Therefore, \eqref{eq:descent-value} gives
\begin{equation}\label{eq:descent-union}
\beta\Bigl(\bar H^{a_1,\dots,a_s}\big|_{B_0}\Bigr)
=\bigl\{\Phi_{s-1}(\uu')+\varepsilon\vv:\ \uu'\in\bar\cH^{a_1,\dots,a_{s-1}},
\ \varepsilon\in\Z_2\bigr\},
\end{equation}
where $\vv$ is the binary vector whose value at the binary coordinate
$(\zz',h')$ is $\wt_H(h')$ read modulo $2$.

It remains to prove that $\vv$ belongs to $\K(\bar H^{a_1,\dots,a_{s-1}})$,
for then $\bar H^{a_1,\dots,a_{s-1}}+\vv=\bar H^{a_1,\dots,a_{s-1}}$ and the
union \eqref{eq:descent-union} collapses to $\bar H^{a_1,\dots,a_{s-1}}$,
which is \eqref{eq:descent}.  For $s=2$, this is trivial, because
$\Z_2^{s-2}$ reduces to a single point, $\wt_H(h')=0$ and $\vv=\zero$.  Let
$s\geq3$ and put $\lambda=2^{\,s-2}-1\in\Z_{2^{s-1}}$, whose binary digits
equal $1$ in the positions $0,\dots,s-3$ and $0$ in the position $s-2$.  By
\eqref{eq:z2s-coordinate-value} applied over $\Z_{2^{s-1}}$,
$$
\phi_{s-1}(\lambda)_{h'}=\lambda_{s-2}+\sum_{\mu=0}^{s-3}h'_\mu\lambda_\mu
=0+\sum_{\mu=0}^{s-3}h'_\mu=\wt_H(h'),
$$
so $\vv=\Phi_{s-1}(\lambda\one)$.  If $\bar H^{a_1,\dots,a_{s-1}}$ is linear,
then its kernel is the whole code, which contains
$\Phi_{s-1}(\lambda\ww_1)=\vv$ because $\ww_1=\one$.  Otherwise,
$\lambda=2^{\,s'-1}-1$ with $s'=s-1$, which is precisely the last vector of
the basis \eqref{eq:z2s-kernel-basis} of $\K(\bar H^{a_1,\dots,a_{s-1}})$
given by Proposition \ref{prop:z2s-kernel-basis}.  In both cases
$\vv\in\K(\bar H^{a_1,\dots,a_{s-1}})$, which completes the proof.
\end{proof}

\begin{corollary}\label{cor:local-rank-bound}
Let $s\geq2$ and let $\bar H^{a_1,\dots,a_s}$ be a nonlinear
$\Z_{2^s}$-linear Hadamard code with $a_1\geq2$, of length $2^t$ and kernel
dimension $\kappa$.  Then, its constant local rank is
\begin{equation}\label{eq:local-rank-general}
\rho_{a_1,\dots,a_s}=\rank\bigl(\bar H^{a_1,\dots,a_{s-1}}\bigr),
\end{equation}
which does not depend on $a_s$; every kernel block has exactly
$2^{\,t-\kappa+1}$ binary coordinates; and
\begin{equation}\label{eq:local-rank-bound}
\rho_{a_1,\dots,a_s}\ \geq\ t-\kappa+2,
\end{equation}
with equality if and only if $\bar H^{a_1,\dots,a_{s-1}}$ is linear, that is,
if and only if $s=2$, or $s=3$ and $a_1=2$.
\end{corollary}

\begin{proof}
By Theorem \ref{thm:z2s-profile-constant}, all the blocks have the same local
rank, and therefore that common value may be computed on the particular
block $B_0$ of
Lemma \ref{lem:zero-block}.  By Theorem \ref{thm:descent}, the punctured code
$\bar H^{a_1,\dots,a_s}|_{B_0}$ is carried by the coordinate bijection $\beta$
onto $\bar H^{a_1,\dots,a_{s-1}}$, so the two codes have the same rank by
Lemma \ref{lem:rank-kernel-invariance}.   By \eqref{eq:span-puncture}, that rank is
$\dim\langle\bar H^{a_1,\dots,a_s}|_{B_0}\rangle
=\dim\langle\bar H^{a_1,\dots,a_s}\rangle|_{B_0}
=\rho_{a_1,\dots,a_s}$, which is \eqref{eq:local-rank-general}.  

For the size of the blocks, recall from Lemma \ref{lem:z2s-kernel-structure}
that there are $2^{\kappa-1}$ blocks. By Theorem
\ref{thm:z2s-profile-constant}, they all have the same size, so each of them
has $2^{t}/2^{\kappa-1}=2^{\,t-\kappa+1}$ binary coordinates.

We now prove the bound.  Through $\beta$, the block $B_0$ is the coordinate
set of $\bar H^{a_1,\dots,a_{s-1}}$, and that code has length
$2^{t'}$ with
$
t'=t-\kappa+1 .
$
Being a Hadamard code of that length, it has $2^{\,t'+1}$ codewords.  Any code
is contained in its linear span, so
$$
2^{\,t'+1}=\bigl|\bar H^{a_1,\dots,a_{s-1}}\bigr|
\leq\bigl|\langle \bar H^{a_1,\dots,a_{s-1}}\rangle\bigr|
=2^{\rank(\bar H^{a_1,\dots,a_{s-1}})}.
$$
Therefore, $\rank(\bar H^{a_1,\dots,a_{s-1}})\geq t'+1=t-\kappa+2$, which is
\eqref{eq:local-rank-bound} by \eqref{eq:local-rank-general}.  Equality holds
if and only if the code and its span have the same cardinality, that is, if
and only if $\bar H^{a_1,\dots,a_{s-1}}$ is linear.

It remains to translate that condition into a condition on $s$ and $a_1$,
using the linearity criteria recalled in Subsection \ref{subsec:z2s-family}.
If $s=2$, then $\bar H^{a_1}$ is by convention the binary linear Hadamard
code, so equality always holds.  If $s=3$, then $\bar H^{a_1,a_2}$ is a
$\Z_4$-linear Hadamard code, which is linear if and only if $a_1\leq2$;
combined with the hypothesis $a_1\geq2$, this gives $a_1=2$.  If $s\geq4$,
then $\bar H^{a_1,\dots,a_{s-1}}$ is a $\Z_{2^{s-1}}$-linear Hadamard code
with $s-1\geq3$, and such a code is linear only when its first parameter
equals $1$; since $a_1\geq2$, equality never holds.
\end{proof}

Corollary \ref{cor:local-rank-bound} determines the local rank whenever the
rank one alphabet lower is known, which by Remark \ref{rem:chain-reduction}
is the case for every code of Tables \ref{tab:small} and \ref{tab:large}
except $\bar H^{2,0,0,0,0}$; the value listed there was obtained with {\sc Magma}. 

The $\Z_2\Z_4\Z_8$-linear Hadamard codes obey a parallel rule, which we record next to
\eqref{eq:local-rank-bound} because it is what makes the comparison
transparent.

\begin{lemma}\label{lem:mixed-local-rank-bound}
Let $H^{t_1,t_2,t_3}$ be a nonlinear $\Z_2\Z_4\Z_8$-linear Hadamard code of
length $2^t$ and kernel dimension $\kappa=t_1+t_2+t_3$.  Then, every kernel
block has exactly $2^{\,t-\kappa+1}$ binary coordinates and, in the notation
of \eqref{eq:rho-def},
\begin{equation}\label{eq:mixed-local-rank-bound}
\rho_1=(t-\kappa+2)+\binom{t_1-1}{2},
\qquad
\rho_0=\rho_1+t_1-1 .
\end{equation}
In particular, $\rho_1=t-\kappa+2$ if and only if $t_1\leq2$, and the profile
is constant and equal to $t-\kappa+2$ if and only if $t_1=1$.
\end{lemma}

\begin{proof}
By Theorem \ref{thm:profile}, the blocks all have $2^{\,2t_1+t_2}$ binary
coordinates, so we only have to check that the exponent is $t-\kappa+1$.
Using \eqref{eq:length-relation} and $\kappa=t_1+t_2+t_3$, we have 
$
t-\kappa+1=(3t_1+2t_2+t_3-1)-(t_1+t_2+t_3)+1=2t_1+t_2 ,
$
as required.  Next, by \eqref{eq:rho-def},
$$
\rho_1-(t-\kappa+2)=t_2+2+\binom{t_1+1}{2}-(2t_1+t_2+1)
=1+\frac{t_1^2+t_1}{2}-2t_1
=\frac{t_1^2-3t_1+2}{2}=\binom{t_1-1}{2},
$$
which is the first identity of \eqref{eq:mixed-local-rank-bound}, and
$\rho_0-\rho_1=t_1-1$ by Theorem \ref{thm:profile}, which is the second one.
Finally, $\binom{t_1-1}{2}=0$ exactly when $t_1-1\leq1$, that is, when
$t_1\leq2$, and the profile is constant exactly when $t_1=1$, again by Theorem
\ref{thm:profile}; in that case $\binom{t_1-1}{2}=0$ and the constant value is
$t-\kappa+2$.
\end{proof}

We close the section by recording the $\Z_8$ values, which we will quote
explicitly, and the reduction that removes the hypothesis $a_1\geq2$.

\begin{corollary}\label{cor:z8-local-rank}
Let $\bar H^{a,b,c}$ be a nonlinear $\Z_8$-linear Hadamard code with $a\geq2$.
Then, all its kernel blocks have size $2^{\,2a+b-1}$ and
\begin{equation}\label{eq:z8-local-rank}
\rho_{a,b,c}=b+1+\binom{a+1}{2}.
\end{equation}
\end{corollary}

\begin{proof}
By \eqref{eq:local-rank-general}, $\rho_{a,b,c}$ is the rank of the
$\Z_4$-linear Hadamard code $\bar H^{a,b}$, whose length is $2^{t'}$ with
$t'+1=2a+b$ by \eqref{eq:z2s-length}.  We distinguish two cases.

If $a=2$, that code is linear, since a $\Z_4$-linear Hadamard code is linear
exactly when its first parameter is at most two.  A linear Hadamard code of
length $2^{t'}$ has rank $t'+1$, so here the rank is $t'+1=b+4$, which is
\eqref{eq:z8-local-rank} because $\binom{a+1}{2}=\binom{3}{2}=3$.

If $a\geq3$, the code $\bar H^{a,b}$ is nonlinear and, by \cite{KV2015}, it is
equivalent to a $\Z_2\Z_4$-linear Hadamard code $H^{U,V}$ with
$\alpha_1\not=0$; being equivalent, the two codes have the same length and the
same kernel dimension.  Comparing \eqref{eq:z2z4-length} and
\eqref{eq:z2z4-invariants} with \eqref{eq:z2s-length} and
\eqref{eq:z2s-kernel} for $s=2$ and $\sigma=1$ gives $2U+V=2a+b$ and  $U+V=1+a+b$,
and subtracting the second identity from the first one gives $U=a-1$. Therefore,
$V=b+2$.  Note that $U=a-1\geq2$ and $V=b+2\geq1$, so $H^{U,V}$ is indeed
nonlinear and \eqref{eq:z2z4-invariants} applies.  Hence,
\begin{equation*}
\begin{aligned}
\rho_{a,b,c}=\rank\bigl(H^{U,V}\bigr)
&=V+2U+\binom{U}{2}=b+2+2(a-1)+\binom{a-1}{2}\\
&=b+\frac{4a+(a-1)(a-2)}{2}=b+\frac{a^2+a+2}{2},
\end{aligned}
\end{equation*}
which equals $b+1+\binom{a+1}{2}$, as claimed.

In both cases, finally, the block size is $2^{\,t-\kappa+1}$ by Corollary
\ref{cor:local-rank-bound}, and here $t+1=3a+2b+c$ by \eqref{eq:z2s-length}
and $\kappa=1+a+b+c$ by \eqref{eq:z2s-kernel}, so
$t-\kappa+1=(3a+2b+c-1)-(1+a+b+c)+1=2a+b-1$.
\end{proof}

\begin{remark}\label{rem:chain-reduction}
The hypothesis $a_1\geq2$ costs nothing.  Applying Theorem
\ref{thm:recalled-equivalence} with $\ell=1$ to a $\Z_{2^{s-1}}$-linear
Hadamard code $\bar H^{u_1,\dots,u_{s-1}}$ with $u_{s-1}\geq1$, we see that it
is permutation equivalent to the $\Z_{2^s}$-linear Hadamard code
$\bar H^{1,u_1-1,u_2,\dots,u_{s-2},u_{s-1}-1}$.  Read backwards, this says
that every $\Z_{2^s}$-linear Hadamard code whose first parameter is $1$, say
$\bar H^{1,b_2,\dots,b_s}$, is permutation equivalent to the
$\Z_{2^{s-1}}$-linear Hadamard code
$\bar H^{\,b_2+1,\,b_3,\dots,b_{s-1},\,b_s+1}$.  The process strictly decreases
$s$ and stops as soon as the first parameter is at least two, and it cannot
run past that point: a code $\bar H^{a_1}$ over $\Z_2$ is linear, and so is a
code $\bar H^{1,a_2}$ over $\Z_4$ by the criterion recalled in Subsection
\ref{subsec:z2s-family}, whereas permutation equivalence preserves
nonlinearity.  Iterating, every nonlinear $\Z_{2^s}$-linear Hadamard code is
therefore permutation equivalent to some $\bar H^{a_1,\dots,a_{s'}}$ with
$2\leq s'\leq s$ and $a_1\geq2$; the case $s'=2$ occurs exactly when the code
is equivalent to a
$\Z_4$-linear one, and then $a_1\geq3$, since the code is nonlinear.  By Lemma
\ref{lem:profile-invariant}, Corollary \ref{cor:local-rank-bound} therefore
describes the local rank of every nonlinear $\Z_{2^s}$-linear Hadamard code.
\end{remark}

\begin{example}\label{ex:descent}
Take the $\Z_{2^s}$-linear Hadamard code $\bar H^{2,0,0}$, of length $2^5$ and kernel dimension $3$, as in Example
\ref{ex:z2s-blocks}.  By Corollary \ref{cor:local-rank-bound}, its blocks have
$2^{\,t-\kappa+1}=2^{3}=8$ coordinates, and the punctured code on a block is the $\Z_4$-linear Hadamard code $\bar H^{2,0}$, of length $2^{\,t-\kappa+1}=8$.  That code is linear, so
$\rho_{2,0,0}=\rank(\bar H^{2,0})=4=t-\kappa+2$ and
\eqref{eq:local-rank-bound} is an equality; formula
\eqref{eq:z8-local-rank} gives the same value, $0+1+\binom{3}{2}=4$.  Take now
$\bar H^{3,0,0}$, of length $2^8$ and kernel dimension $4$.  The punctured
code on a block is $\bar H^{3,0}$, of length $2^{8-4+1}=2^5$, which is
nonlinear of rank $7$; hence $\rho_{3,0,0}=7>6=t-\kappa+2$, and
\eqref{eq:z8-local-rank} agrees once more, since
$0+1+\binom{4}{2}=7$.  These are exactly the two situations of Corollary
\ref{cor:local-rank-bound}: equality for $s=3$ and $a_1=2$, strict inequality
as soon as $a_1\geq3$.
\end{example}

\section{The complete classification}\label{sec:cross-family}

We can now compare the $\Z_2\Z_4\Z_8$-linear Hadamard codes $H^{t_1,t_2,t_3}$ with the two families of Subsections \ref{subsec:z2z4-family} and \ref{subsec:z2s-family}.  The comparison splits into two cases according to the value of $t_1$, because, by Theorem
\ref{thm:profile}, the profile of $H^{t_1,t_2,t_3}$ is nonconstant exactly
when $t_1\geq2$, whereas the profiles of the two comparison families are
always constant by Proposition \ref{prop:z2z4-profile} and Theorem
\ref{thm:z2s-profile-constant}.  The case $t_1\geq2$ is therefore immediate,
and the case $t_1=1$ needs the finer information supplied by Section
\ref{sec:descent}.

\subsection{The codes with $t_1\geq2$}\label{subsec:t1-at-least-2}

\begin{theorem}\label{thm:t1-at-least-2}
Let $H^{t_1,t_2,t_3}$ be a $\Z_2\Z_4\Z_8$-linear Hadamard code with
$t_1\geq2$.  Then, $H^{t_1,t_2,t_3}$ is not equivalent to any $\Z_4$-linear
Hadamard code, nor to any $\Z_2\Z_4$-linear Hadamard code, nor to any
$\Z_{2^s}$-linear Hadamard code with $s\geq2$, of the same length.
\end{theorem}

\begin{proof}
Since $t_1\geq2$, we have $(t_1,t_2)\not=(1,0)$, so the code $H^{t_1,t_2,t_3}$
is nonlinear by Theorem \ref{thm:recalled-mixed}, and by Theorem
\ref{thm:profile} its profile takes the two values $\rho_0$ and $\rho_1$ with
$\rho_0-\rho_1=t_1-1\geq1$.  Both values occur, because both multiplicities in
\eqref{eq:profile-formula} are positive; hence the profile is a nonconstant
multiset.

Let now $D$ be a comparison code of the same length, that is, a $\Z_4$-linear,
a $\Z_2\Z_4$-linear or a $\Z_{2^s}$-linear Hadamard code of length $2^t$.  We
show in every case that the profile of $D$ is constant.  If $D$ is linear,
then all its blocks are singletons and $\cR(D)=\{\!\{1^{[2^t]}\}\!\}$ by Lemma
\ref{lem:linear-profile}, which is constant.  If $D$ is a nonlinear
$\Z_2\Z_4$-linear Hadamard code, its profile is constant by Proposition
\ref{prop:z2z4-profile}.  If $D$ is a nonlinear $\Z_{2^s}$-linear Hadamard
code, its profile is constant by Theorem \ref{thm:z2s-profile-constant}.
Finally, a $\Z_4$-linear Hadamard code is the case $s=2$.

Since the profile is an equivalence invariant by Lemma
\ref{lem:profile-invariant}, the two codes are nonequivalent.
\end{proof}

\subsection{The codes with $t_1=1$}\label{subsec:t1-one}

When $t_1=1$, the profile of $H^{1,t_2,t_3}$ is constant by
\eqref{eq:profile-t1-one}, so it no longer separates automatically.  The
results of Section \ref{sec:descent} nevertheless settle almost all the cases
at once, and uniformly in $s$.  Indeed, by Lemma
\ref{lem:mixed-local-rank-bound}, the constant value of the profile of
$H^{1,t_2,t_3}$ is $t-\kappa+2$, whereas by Corollary \ref{cor:local-rank-bound}, a nonlinear
$\Z_{2^s}$-linear Hadamard code with $a_1\geq2$ attains that minimum only in
two very restricted situations.  Whenever the comparison code exceeds the
minimum, the two constant profiles are different and the codes are
nonequivalent; only the two extremal situations require a separate argument,
and they are handled by the rank.

We begin by observing that, once the length and the dimension of the kernel of
the two codes agree, the parameters of the comparison code are completely
determined.

\begin{lemma}\label{lem:t1-one-parameters}
Let $H^{1,t_2,t_3}$ be nonlinear, that is, with $t_2\geq1$, and let its length
be $2^t$ and its kernel dimension be $\kappa=1+t_2+t_3$.
\begin{enumerate}[label=\textup{(\roman*)}]
\item A nonlinear $\Z_2\Z_4$-linear Hadamard code $H^{U,V}$ of the same length
and the same kernel dimension exists if and only if $t_3\geq2$, and then
$(U,V)=(t_2+2,\,t_3-1)$.
\item A nonlinear $\Z_8$-linear Hadamard code $\bar H^{a,b,c}$ with $a\geq2$
of the same length and the same kernel dimension satisfies $b=t_2+3-2a$ and
$c=t_3+a-3$.
\end{enumerate}
\end{lemma}

\begin{proof}
That $t_2\geq1$ follows from Theorem \ref{thm:recalled-mixed}: the code is
nonlinear exactly when $(t_1,t_2)\not=(1,0)$, and here $t_1=1$.  Note also
that, by \eqref{eq:length-relation} with $t_1=1$,
\begin{equation}\label{eq:t1one-length}
t+1=3+2t_2+t_3 .
\end{equation}

We prove item (i).  Equating the lengths of $H^{1,t_2,t_3}$ and $H^{U,V}$ means
equating the two expressions for $t+1$, namely \eqref{eq:t1one-length} and
\eqref{eq:z2z4-length}, which gives $3+2t_2+t_3=2U+V$.  Equating the kernel
dimensions and using \eqref{eq:z2z4-invariants} gives $1+t_2+t_3=U+V$.
Subtracting the second identity from the first one gives $U=t_2+2$, and then
$V=(1+t_2+t_3)-U=t_3-1$.  A $\Z_2\Z_4$-linear Hadamard code with
$\alpha_1\not=0$ has $V\geq1$, so such a code exists exactly when
$t_3-1\geq1$, that is, when $t_3\geq2$; and it is then nonlinear, because
$U=t_2+2\geq3$ and linearity would require $U=1$.

We prove item (ii).  Since $a\geq2$, we have $\sigma=1$ by \eqref{eq:sigma}, so
\eqref{eq:z2s-length} and \eqref{eq:z2s-kernel} read $t+1=3a+2b+c$ and
$\kappa=1+a+b+c$.  Equating them with \eqref{eq:t1one-length} and with
$\kappa=1+t_2+t_3$ gives $3+2t_2+t_3=3a+2b+c$ and $1+t_2+t_3=1+a+b+c.$ The second identity simplifies to $t_2+t_3=a+b+c$.  Subtracting it from the first one gives $3+t_2=2a+b$, so $b=t_2+3-2a$. Substituting this into $t_2+t_3=a+b+c$ gives $c=t_3+a-3$.
\end{proof}

\begin{proposition}\label{prop:t1-one-z2z4}
Let $H^{1,t_2,t_3}$ be nonlinear.  Then, $H^{1,t_2,t_3}$ is not equivalent to
any $\Z_4$-linear Hadamard code, nor to any $\Z_2\Z_4$-linear Hadamard code,
of the same length.
\end{proposition}

\begin{proof}
A linear comparison code is separated from $H^{1,t_2,t_3}$ by linearity. Let
therefore $H^{U,V}$ be a nonlinear $\Z_2\Z_4$-linear Hadamard code of the same
length.  If its kernel dimension differs from that of $H^{1,t_2,t_3}$, the two
codes are nonequivalent by Lemma \ref{lem:rank-kernel-invariance} and we are
done.  Assume then that the kernel dimensions agree.  By item (i) of Lemma
\ref{lem:t1-one-parameters}, we have $U=t_2+2$ and $V=t_3-1$, so, by
\eqref{eq:z2z4-invariants},
$\rank\bigl(H^{U,V}\bigr)=V+2U+\binom{U}{2}=(t_3-1)+2(t_2+2)+\binom{t_2+2}{2},$
whereas \eqref{eq:rank-binomial} with $t_1=1$ gives
$$
\rank\bigl(H^{1,t_2,t_3}\bigr)=(t_3-1)+4+3t_2+\binom{t_2+1}{2}.
$$
Subtracting the second expression from the first one, the terms $t_3-1$ cancel
and we get
$$
2t_2+4-4-3t_2+\binom{t_2+2}{2}-\binom{t_2+1}{2}
=-t_2+\frac{(t_2+2)(t_2+1)-(t_2+1)t_2}{2}
=-t_2+(t_2+1)=1 .
$$
The two ranks therefore differ by exactly one, so they are different, and the
two codes are nonequivalent by Lemma \ref{lem:rank-kernel-invariance}.
Finally, every $\Z_4$-linear Hadamard code is equivalent to a
$\Z_2\Z_4$-linear one with $\alpha_1\not=0$ and $\alpha_2\not=0$ by
\cite{KV2015}, so it is covered by the preceding argument.
\end{proof}

\begin{theorem}\label{thm:t1-one-z2s}
Let $H^{1,t_2,t_3}$ be nonlinear and let $\bar H$ be a nonlinear
$\Z_{2^s}$-linear Hadamard code of the same length $2^t$, with $s\geq2$.
Unless $t_2=1$ and $\bar H$ is permutation equivalent to $\bar H^{2,0,t-5}$,
the codes $H^{1,t_2,t_3}$ and $\bar H$ are nonequivalent.
\end{theorem}

\begin{proof}
If the two codes have different kernel dimensions, they
are nonequivalent by Lemma \ref{lem:rank-kernel-invariance} and there is
nothing to prove; so assume that they have the same kernel dimension, say
$\kappa$.  By Remark \ref{rem:chain-reduction}, we may replace $\bar H$ by a
permutation equivalent code and assume that $\bar H=\bar H^{a_1,\dots,a_{s'}}$
with $2\leq s'\leq s$ and $a_1\geq2$; this changes neither the length, nor the
kernel dimension, nor the profile, by Lemmas \ref{lem:rank-kernel-invariance}
and \ref{lem:profile-invariant}.

We compare the two profiles.  By Theorem \ref{thm:profile}, the code
$H^{1,t_2,t_3}$ has $2^{\kappa-1}$ kernel blocks, by Lemma
\ref{lem:mixed-local-rank-bound}, they have size $2^{\,t-\kappa+1}$, and its
profile is constant and equal to $t-\kappa+2$, because $t_1=1$.  By Lemma
\ref{lem:z2s-kernel-structure}, the code $\bar H$ also has $2^{\kappa-1}$
kernel blocks, by Corollary \ref{cor:local-rank-bound}, they have size
$2^{\,t-\kappa+1}$, and its profile is constant with some value
$\rho\geq t-\kappa+2$, with equality if and only if $s'=2$, or $s'=3$ and
$a_1=2$.

If neither of those two equalities occurs, then $\rho>t-\kappa+2$, the two
constant profiles take different values, the two multisets are different, and
the codes are nonequivalent by Lemma \ref{lem:profile-invariant}.  Two cases
remain.

Suppose first that $s'=2$.  Then, $\bar H$ is a $\Z_4$-linear Hadamard code, and
Proposition \ref{prop:t1-one-z2z4} shows that it is not equivalent to
$H^{1,t_2,t_3}$.

Suppose finally that $s'=3$ and $a_1=2$; then, $\bar H=\bar H^{2,b,c}$ for some $b$
and $c$.  By item (ii) of Lemma \ref{lem:t1-one-parameters} with $a=2$, we get
$b=t_2+3-4=t_2-1$ and $c=t_3+2-3=t_3-1$; in particular, such a comparison code
exists only when $t_2\geq1$, which holds because $H^{1,t_2,t_3}$ is nonlinear.
We compare the ranks.  Substituting $(a,b,c)=(2,t_2-1,t_3-1)$ into
\eqref{eq:z8-rank}, and using
$(2^4-2\cdot2^3+35\cdot2^2+14\cdot2)/24=168/24=7$ and
$b(a^2+a+b+1)/2=(t_2-1)(t_2+6)/2$,
we obtain
$$
\rank\bigl(\bar H^{2,t_2-1,t_3-1}\bigr)
=7+\frac{t_2^2+5t_2-6}{2}+(t_3-1)+1
=7+\frac{t_2^2+5t_2-6}{2}+t_3 ,
$$
whereas, as computed in the proof of Proposition \ref{prop:t1-one-z2z4},
$$
\rank\bigl(H^{1,t_2,t_3}\bigr)=(t_3-1)+4+3t_2+\binom{t_2+1}{2}
=t_3+3+3t_2+\frac{t_2^2+t_2}{2}.
$$
Subtracting the second expression from the first one, the terms $t_3$ cancel
and
$$
7-3-3t_2+\frac{(t_2^2+5t_2-6)-(t_2^2+t_2)}{2}
=4-3t_2+\frac{4t_2-6}{2}=4-3t_2+2t_2-3=1-t_2 .
$$
Hence, the two ranks differ whenever $t_2\not=1$, and in that case the codes
are nonequivalent by Lemma \ref{lem:rank-kernel-invariance}.  If $t_2=1$, then
$b=0$ and $c=t_3-1$, and \eqref{eq:t1one-length} with $t_2=1$ gives
$t+1=3+2+t_3$, that is, $t_3=t-4$ and $c=t-5$.  The comparison code is thus
$\bar H^{2,0,t-5}$, which is the excluded case.
\end{proof}

The pair left open by Theorem \ref{thm:t1-one-z2s} is therefore $H^{1,1,t-4}$
and $\bar H^{2,0,t-5}$ with $t\geq5$, and these two codes share a remarkable
amount of structure.  Both have kernel dimension $t-2$ and rank $t+3$, by
\eqref{eq:z8-five}; both have $2^{\,t-3}$ kernel blocks of size $8$; and, by
\eqref{eq:profile-t1-one} and Corollary \ref{cor:z8-local-rank}, the local
rank on every single block is $4$ for both.  In the language of Section
\ref{sec:descent}, both are extremal, their common local rank attaining the
lower bound $t-\kappa+2=4$: on the $\Z_2\Z_4\Z_8$ side because $t_1=1$, and on the
other side because the punctured code on a block is the linear Hadamard
code of length $8$.  Hence, neither the rank nor the dimension of the kernel,
nor the profile can separate them, and we must refine the invariant.

\subsection{A two-block refinement}\label{subsec:two-block}

The profile records the dimension of the span punctured on one block at a
time.  The natural refinement is to puncture on the union of two blocks, and
that is enough to settle the remaining family.

\begin{definition}\label{def:mu}
Let $C\subseteq\Z_2^N$ be a binary code with $\zero\in C$, and assume that $C$
has at least two kernel blocks.  We define
\begin{equation}\label{eq:mu}
\mu(C)=\max\bigl\{\dim\langle C\rangle|_{B\cup B'}:\ B,B'\in\cB(\K(C)),
\ B\not=B'\bigr\}.
\end{equation}
\end{definition}

\begin{lemma}\label{lem:mu-invariant}
The integer $\mu(C)$ is invariant under binary code equivalence.
\end{lemma}

\begin{proof}
Let $D=\yy+\pi(C)$ be a code with $\zero\in C\cap D$.  By Lemma
\ref{lem:profile-invariant}, the permutation $\pi$ maps the kernel blocks of
$C$ bijectively onto those of $D$; hence it maps unordered pairs of distinct
kernel blocks of $C$ onto unordered pairs of distinct kernel blocks of $D$,
and this correspondence between pairs is a bijection.

We now compare the restricted spaces.  By Lemma
\ref{lem:rank-kernel-invariance}, we have $\langle D\rangle=\pi(\langle
C\rangle)$, so \eqref{eq:relabel} applies with $V=\langle C\rangle$.  Taking
$S=B\cup B'$ for two distinct blocks $B,B'\in\cB(\K(C))$, and using
$\pi(B\cup B')=\pi(B)\cup\pi(B')$, we obtain
$$
\dim\langle D\rangle|_{\pi(B)\cup\pi(B')}
=\dim\langle C\rangle|_{B\cup B'} .
$$
As the pair $\{B,B'\}$ runs over the unordered pairs of distinct kernel blocks
of $C$, the pair $\{\pi(B),\pi(B')\}$ runs over those of $D$.  Taking the
maximum over all pairs on both sides therefore gives $\mu(D)=\mu(C)$.
\end{proof}

By Lemma \ref{lem:evaluation-principle}, the quantity
$\dim\langle C\rangle|_{B\cup B'}$ is the dimension of the sum of the
two spaces of coordinate functions attached to $B$ and to $B'$; this is how we
will compute it.  The next lemma shows how $\mu$ behaves under a Plotkin
extension, and it is what makes the computation finite: it reduces the whole
infinite family to the single shortest member.

\begin{lemma}\label{lem:mu-plotkin}
Let $C\subseteq\Z_2^N$ be a binary code with $\zero\in C$ and $\one\in\K(C)$,
having at least two kernel blocks, and let $C^{+}$ be its Plotkin extension
\eqref{eq:plotkin-extension}.  Then,
\begin{equation}\label{eq:mu-recurrence}
\mu(C^{+})=\max\Bigl\{\mu(C),\ 1+\max_{B\in\cB(\K(C))}\rho_C(B)\Bigr\}.
\end{equation}
\end{lemma}

\begin{proof}
By Lemma \ref{lem:plotkin-blocks}, the kernel blocks of $C^{+}$ are the sets
$B^{0}$ and $B^{1}$, with $B$ running over the kernel blocks of $C$, and
$\langle C^{+}\rangle$ is described by \eqref{eq:plotkin-kernel-span}.  A pair
of distinct blocks of $C^{+}$ is therefore of the form
$\{B^{\varepsilon},B'^{\varphi}\}$ with $B,B'$ blocks of $C$ and
$\varepsilon,\varphi\in\{0,1\}$, the two members being distinct.  Two kinds of
pair occur, and we compute the punctured dimension in each of them.

Consider first a pair with $B=B'$ and $\varepsilon\not=\varphi$, that is, the
pair $\{B^{0},B^{1}\}$.  Restricted to $B^{0}\cup B^{1}$, the diagonal part
$\{(\yy,\yy):\yy\in\langle C\rangle\}$ of \eqref{eq:plotkin-kernel-span} gives
the space $\{(\yy|_B,\yy|_B)\}$, which is isomorphic to $\langle C\rangle|_B$
and therefore has dimension $\rho_C(B)$.  The extra vector $(\zero,\one)$
restricts to the vector that vanishes on $B^{0}$ and is the all-one vector on
$B^{1}$; that vector is not of the form $(\yy|_B,\yy|_B)$, because $B\not=\emptyset$.  Hence, it adds exactly one dimension, and
$$
\dim\langle C^{+}\rangle|_{B^{0}\cup B^{1}}=\rho_C(B)+1 .
$$

Consider next a pair with $B\not=B'$.  Restricted to
$B^{\varepsilon}\cup B'^{\varphi}$, the diagonal part gives a space isomorphic
to $\langle C\rangle|_{B\cup B'}$, of dimension
$\dim\langle C\rangle|_{B\cup B'}$, and we must decide whether the extra
vector $(\zero,\one)$ adds a dimension.  If $\varepsilon=\varphi$, then its
restriction is constant, equal to $\zero$ when $\varepsilon=0$ and to the
all-one vector when $\varepsilon=1$; both of these already belong to
$\langle C\rangle|_{B\cup B'}$, because $\one\in\K(C)\subseteq\langle
C\rangle$.  If $\varepsilon\not=\varphi$, then its restriction is constant
equal to $\varepsilon$ on $B$ and constant equal to $\varphi$ on $B'$.  Since
$B$ and $B'$ are different kernel blocks, there exists $\xx\in\K(C)$ taking
different constant values on $B$ and on $B'$. Adding $\one$ to $\xx$ if
necessary, we may assume that those two values are $\varepsilon$ and
$\varphi$, and then the restriction of $\xx$ to $B\cup B'$ is exactly the
vector under consideration, which therefore lies in $\langle C\rangle|_{B\cup
B'}$ as well.  In both cases, no new dimension is added, and
$$
\dim\langle C^{+}\rangle|_{B^{\varepsilon}\cup B'^{\varphi}}
=\dim\langle C\rangle|_{B\cup B'} .
$$
Taking the maximum over the two kinds of pairs gives
\eqref{eq:mu-recurrence}.
\end{proof}

\begin{lemma}\label{lem:mu-base}
At binary length $32$, we have $\mu(H^{1,1,1})=6$ and
$\mu(\bar H^{2,0,0})=7$.
\end{lemma}

\begin{proof}
Both codes have four kernel blocks and all their local ranks are equal to
$4$: for $H^{1,1,1}$ by Theorem \ref{thm:profile} with $t_1=t_2=t_3=1$, and
for $\bar H^{2,0,0}$ by Lemma \ref{lem:z2s-kernel-structure} and Corollary
\ref{cor:z8-local-rank} with $(a,b,c)=(2,0,0)$.  In each case, there are
therefore $\binom{4}{2}=6$ pairs of distinct blocks to examine.  For each code,
we list the four local spaces of coordinate functions and then compute the dimensions
of the six sums. By Lemma \ref{lem:evaluation-principle}, the dimension of the
sum of two local spaces is the punctured rank on the union of the
corresponding two blocks.  

We start with the code $H^{1,1,1}$, where $t_1=t_2=t_3=1$.  Write a codeword as $\epsilon\ww_0+x_1\ww_1+y_1\vv_1$
with $\epsilon\in\Z_2$, $x_1=a_1+2b_1+4c_1\in\Z_8$ and
$y_1=d_1+2e_1\in\Z_4$, as in \eqref{eq:message-digits}, and write
$\cV_{u,v}$ for the space of coordinate functions of the block $B_{u,v}$, the
labels $u$ and $v$ being single bits here.  We specialise Proposition
\ref{prop:recalled-local-bases}.

For $u=0$ and $v=0$, formula \eqref{eq:B00-basis} lists $\epsilon$, $a_1$,
$d_1$ and $b_1$; there is no product $a_ia_k$, since $t_1=1$.  For $u=0$ and
$v=1$, we have $R=\{1\}$, so $F_v=\epsilon+\eta_R+\sym_2(d_R)=\epsilon+e_1$
and $g_1=b_1+a_1\delta_R=b_1+a_1d_1$, and formula \eqref{eq:B0v-basis} lists
$a_1$, $d_1$, $\epsilon+e_1$ and $b_1+a_1d_1$.  For $u=1$, we have $S=\{1\}$
and $i_0=1$, so there is no function $\Psi_i$ and no product $a_ia_k$, and
formula \eqref{eq:Buv-basis} lists $a_1$, $d_1$, $Q_0$ and $F_0$. Evaluating
\eqref{eq:Q0F0} with $|S|=1$, which makes $\sym_2(a_S)$, $\sym_2(b_S)$ and
$\sym_4(a_S)$ vanish, and using that $t_2=1$ forces $|R|\leq1$ and hence
$\sym_2(d_R)=0$, we get $Q_0=b_1+\delta_R$ and
$F_0=\epsilon+c_1+\eta_R+b_1\delta_R$, that is, $Q_0=b_1$ and
$F_0=\epsilon+c_1$ when $v=0$, and $Q_0=b_1+d_1$ and
$F_0=\epsilon+c_1+e_1+b_1d_1$ when $v=1$.  Altogether,
$$
\begin{aligned}
\cV_{0,0}&=\Span\{\epsilon,\ a_1,\ b_1,\ d_1\},&
\cV_{0,1}&=\Span\{a_1,\ d_1,\ \epsilon+e_1,\ b_1+a_1d_1\},\\
\cV_{1,0}&=\Span\{a_1,\ d_1,\ b_1,\ \epsilon+c_1\},&
\cV_{1,1}&=\Span\{a_1,\ d_1,\ b_1,\ \epsilon+c_1+e_1+b_1d_1\},
\end{aligned}
$$
where in $\cV_{1,1}$ we replaced the generator $Q_0=b_1+d_1$ by $b_1$, which
is legitimate because $d_1$ is also a generator.  Each of the four spaces has
dimension $4$, in accordance with \eqref{eq:profile-t1-one}.

We now compute the six sums.  Adding $\cV_{0,0}$ to $\cV_{0,1}$ produces
$\epsilon$ and $\epsilon+e_1$, hence $e_1$, and $b_1$ and $b_1+a_1d_1$, hence
$a_1d_1$. Adding $\cV_{0,0}$ to $\cV_{1,0}$ produces $c_1$; and adding
$\cV_{0,0}$ to $\cV_{1,1}$ produces $c_1+e_1+b_1d_1$.  Adding $\cV_{0,1}$ to
$\cV_{1,0}$ produces $b_1$, hence $a_1d_1$, together with $\epsilon+c_1$.
Adding $\cV_{0,1}$ to $\cV_{1,1}$ produces $b_1$, hence $a_1d_1$, together
with $(\epsilon+c_1+e_1+b_1d_1)+(\epsilon+e_1)=c_1+b_1d_1$; and adding
$\cV_{1,0}$ to $\cV_{1,1}$ produces
$(\epsilon+c_1)+(\epsilon+c_1+e_1+b_1d_1)=e_1+b_1d_1$.  Hence,
$$
\begin{aligned}
\cV_{0,0}+\cV_{0,1}&=\Span\{\epsilon,a_1,b_1,d_1,e_1,a_1d_1\},&&\dim=6,\\
\cV_{0,0}+\cV_{1,0}&=\Span\{\epsilon,a_1,b_1,d_1,c_1\},&&\dim=5,\\
\cV_{0,0}+\cV_{1,1}&=\Span\{\epsilon,a_1,b_1,d_1,c_1+e_1+b_1d_1\},&&\dim=5,\\
\cV_{0,1}+\cV_{1,0}&=\Span\{a_1,d_1,b_1,a_1d_1,\epsilon+e_1,\epsilon+c_1\},
  &&\dim=6,\\
\cV_{0,1}+\cV_{1,1}&=\Span\{a_1,d_1,b_1,a_1d_1,\epsilon+e_1,c_1+b_1d_1\},&&\dim=6,\\
\cV_{1,0}+\cV_{1,1}&=\Span\{a_1,d_1,b_1,\epsilon+c_1,e_1+b_1d_1\},&&\dim=5,
\end{aligned}
$$
Each of the six generating sets is linearly independent, because every
function listed in it contains a monomial occurring in none of the others; in
the fourth set, for instance, these monomials are $a_1$, $d_1$, $b_1$,
$a_1d_1$, $e_1$ and $c_1$.  The largest of the six dimensions is $6$, so
$\mu(H^{1,1,1})=6$.

We pass to the code $\bar H^{2,0,0}$.
Here, $s=3$, $m=2$ and both rows have level $1$, so by \eqref{eq:z2s-columns},
the matrix $\bar A^{2,0,0}$ is the $2\times8$ matrix whose columns are
$(1,j)^T$ with $j\in\Z_8$.  The message is $(x_1,x_2)\in\Z_8^2$, and we
expand $x_1=p_0+2p_1+4p_2$ and $x_2=q_0+2q_1+4q_2$ in binary; these six digits
are the independent binary variables of the present computation.  By
\eqref{eq:z2s-symbol}, the entry at the additive coordinate $(1,j)$ is
$W=x_1+x_2j$, and by \eqref{eq:z2s-coordinate-value}, the binary coordinate
$(j,h)$, with $h=(h_0,h_1)$, carries the value $W_2+h_0W_0+h_1W_1$.

By Lemma \ref{lem:z2s-kernel-structure} with $\sigma=1$, the block of $(j,h)$
is determined by the pair $(\chi,\tau)$ with $\chi=j\bmod2$ and
$\tau=h_0+h_1$, as we already observed in Example \ref{ex:z2s-blocks}.  Let us
parametrise a block: writing $j=\chi+2\xi$ with $\xi=\xi_0+2\xi_1\in\Z_4$, and
$h=(h_0,h_0+\tau)$, the block with label $(\chi,\tau)$ is described by the
three free bits $\xi_0$, $\xi_1$ and $h_0$, so it has $8$ binary coordinates,
as it must.

Put $w=x_1+\chi x_2$, with binary digits $w_0,w_1,w_2$. Thus, $w=x_1$ when
$\chi=0$ and, when $\chi=1$, Lemma \ref{lem:carries} applied to the sum of
$x_1$ and $x_2$ gives
$$
w_0=p_0+q_0,\qquad w_1=p_1+q_1+p_0q_0,\qquad
w_2=p_2+q_2+p_1q_1+p_0q_0p_1+p_0q_0q_1 .
$$
Since $W=x_1+x_2(\chi+2\xi)=w+2x_2\xi$ and, modulo $8$,
$$
2x_2\xi=2x_2(\xi_0+2\xi_1)\equiv 2q_0\xi_0+4\bigl(q_0\xi_1+q_1\xi_0\bigr),
$$
adding $2x_2\xi$ to $w$ leaves the bottom digit unchanged, adds $q_0\xi_0$ to
the middle digit with a carry $w_1q_0\xi_0$ into the top digit, and adds
$q_0\xi_1+q_1\xi_0$ to the top digit; that is,
$$
W_0=w_0,\qquad W_1=w_1+q_0\xi_0,\qquad
W_2=w_2+q_0\xi_1+q_1\xi_0+w_1q_0\xi_0 .
$$
Substituting these three digits and $h_1=h_0+\tau$ into
$W_2+h_0W_0+h_1W_1$ and grouping the terms, we obtain the following
coordinate function for $(j,h)$:
$$
w_2+q_0\xi_1+q_1\xi_0+w_1q_0\xi_0+\tau\bigl(w_1+q_0\xi_0\bigr)
 +h_0\bigl(w_0+w_1+q_0\xi_0\bigr).
$$
This expression is a multilinear polynomial in the three free parameters
$\xi_0$, $\xi_1$ and $h_0$, so Lemma \ref{lem:mobius} applies and the local
space is spanned by the coefficients of $1$, $\xi_0$, $\xi_1$, $h_0$ and
$h_0\xi_0$, namely
$$
w_2+\tau w_1,\qquad q_1+w_1q_0+\tau q_0,\qquad q_0,\qquad w_0+w_1,\qquad q_0 .
$$
The fifth coefficient repeats the third one.  Since $q_0$ belongs to the
space, we may add multiples of $q_0$ to the second and to the fourth
coefficient.  For $\chi=0$, this gives,
using $w_1q_0=p_1q_0$ and $w_0+w_1=p_0+p_1$,
$$
\cV_{0,\tau}=\Span\{p_2+\tau p_1,\ q_1+p_1q_0,\ q_0,\ p_0+p_1\},
$$
and, for $\chi=1$, using $w_1q_0=p_1q_0+q_1q_0+p_0q_0$ and
$w_0+w_1=p_0+p_1+q_1+p_0q_0+q_0$,
$$
\cV_{1,\tau}=\Span\{w_2+\tau w_1,\ q_1+p_1q_0+q_1q_0+p_0q_0,\ q_0,\
 p_0+p_1+q_1+p_0q_0\} .
$$
All four spaces have dimension $4$, in agreement with Corollary
\ref{cor:z8-local-rank}; for instance, the four generators of $\cV_{1,\tau}$
contain respectively the private monomials $q_2$, $q_1q_0$, $q_0$ and $p_0$.

We now compute the six sums.  The two blocks with $\chi=0$ give
$$
\cV_{0,0}+\cV_{0,1}=\Span\{p_2,\ p_1,\ q_1+p_1q_0,\ q_0,\ p_0+p_1\}
=\Span\{p_0,p_1,p_2,q_0,q_1+p_1q_0\},
$$
of dimension $5$, the two generators $p_2$ and $p_2+p_1$ producing $p_1$ and
hence $p_0$.  The two blocks with $\chi=1$ give
$\cV_{1,0}+\cV_{1,1}=\cV_{1,0}+\Span\{w_1\}$, whose dimension is $5$ as soon
as $w_1=p_1+q_1+p_0q_0$ does not lie in $\cV_{1,0}$.  Suppose that it did.
The monomial $q_2$ occurs in $w_2$ and in no other generator of $\cV_{1,0}$,
so $w_2$ cannot occur in the combination. The monomial $p_0$ then occurs only
in $p_0+p_1+q_1+p_0q_0$, which is excluded as well; and $w_1$ is not a
combination of $q_1+p_1q_0+q_1q_0+p_0q_0$ and $q_0$, because the monomial
$p_1$ occurs in $w_1$ and in neither of them.  Hence, $w_1\notin\cV_{1,0}$ and
the sum has dimension $5$.

For the four remaining pairs, one block has $\chi=0$ and the other has
$\chi=1$.  Starting from $\cV_{0,\tau}$, of dimension $4$, we add the three
generators of $\cV_{1,\tau'}$ other than $q_0$, which is already present.
Adding $p_0+p_1+q_1+p_0q_0$ and reducing by $p_0+p_1$ and by $q_1+p_1q_0$
leaves $p_0q_0+p_1q_0$, which is new because of the monomial $p_0q_0$. Adding
$q_1+p_1q_0+q_1q_0+p_0q_0$ and reducing by $q_1+p_1q_0$ leaves
$q_1q_0+p_0q_0$, which is new because of the monomial $q_1q_0$; and the last
generator is new because it involves the variable $q_2$, which occurs nowhere
else.  Hence, each of these four sums has dimension $4+3=7$.

Therefore, the six two-block ranks of $\bar H^{2,0,0}$ are $5,5,7,7,7,7$, and
$\mu(\bar H^{2,0,0})=7$.
\end{proof}

\begin{corollary}\label{cor:exceptional-family}
For every $t\geq5$, we have $\mu\bigl(H^{1,1,t-4}\bigr)=6$ and
$\mu\bigl(\bar H^{2,0,t-5}\bigr)=7$, and consequently the codes $H^{1,1,t-4}$
and $\bar H^{2,0,t-5}$ are nonequivalent.
\end{corollary}

\begin{proof}
We argue by induction on $t$.  For $t=5$, the two codes are $H^{1,1,1}$ and
$\bar H^{2,0,0}$, and the two values of $\mu$ are given by Lemma
\ref{lem:mu-base}.

Assume the statement for some $t\geq5$, and pass from $t$ to $t+1$.  On the
$\Z_2\Z_4\Z_8$ side, going from $H^{1,1,t-4}$ to $H^{1,1,t-3}$ increases $t_3$ by one,
which, by the construction recalled in Subsection
\ref{subsec:mixed-recalled}, amounts to applying the duplication that appends
a row of order $2$. As explained in the last part of the proof of Proposition
\ref{prop:z2z4-profile}, and using $\phi_2(2)=(1,1)$ and
$\phi_3(4)=(1,1,1,1)$, this is the Plotkin extension
\eqref{eq:plotkin-extension} up to a permutation of coordinates.  On the other
side, going from $\bar H^{2,0,t-5}$ to $\bar H^{2,0,t-4}$ increases $a_3$ by
one, which by \eqref{eq:z2s-columns} doubles the set of columns, with the new row
taking the values $0$ and $2^{s-1}=4$ on the two copies. Since
$\phi_3(w+4)=\phi_3(w)+\one$ by Lemma \ref{lem:gray-shift}, this is again the
Plotkin extension \eqref{eq:plotkin-extension} up to a permutation of
coordinates.  In both cases the all-one vector belongs to the kernel, by
Theorem \ref{thm:recalled-mixed} and by Proposition
\ref{prop:z2s-kernel-basis} respectively, so Lemma \ref{lem:mu-plotkin}
applies.

It remains to evaluate the right-hand side of \eqref{eq:mu-recurrence}.  By
Theorem \ref{thm:profile}, the profile of $H^{1,1,t-4}$ is constant with value
$t_2+3=4$, and by Corollary \ref{cor:z8-local-rank}, the profile of
$\bar H^{2,0,t-5}$ is constant with value $0+1+\binom{3}{2}=4$; neither value
depends on the last parameter, so the maximum local rank is $4$ at every step.
Since $1+4=5$ is smaller than both $6$ and $7$, the recurrence
\eqref{eq:mu-recurrence} gives
$\mu\bigl(H^{1,1,t-3}\bigr)=\max\{6,5\}=6$ and 
$\mu\bigl(\bar H^{2,0,t-4}\bigr)=\max\{7,5\}=7,$
which is the statement for $t+1$.  The induction is complete, and the two
codes are nonequivalent for every $t\geq5$ because $\mu$ is an equivalence
invariant by Lemma \ref{lem:mu-invariant}.
\end{proof}

\subsection{The classification theorem}\label{subsec:final}

Everything is now in place.

\begin{theorem}\label{thm:cross-family}
Let $H^{t_1,t_2,t_3}$ be a nonlinear $\Z_2\Z_4\Z_8$-linear Hadamard code of
length $2^t$.  Then, for every $s\geq2$, the code $H^{t_1,t_2,t_3}$ is not
equivalent to any $\Z_4$-linear Hadamard code, nor to any $\Z_2\Z_4$-linear
Hadamard code, nor to any $\Z_{2^s}$-linear Hadamard code, of the same length
$2^t$.
\end{theorem}

\begin{proof}
Since $t_1\geq1$, exactly one of the two cases $t_1\geq2$ and $t_1=1$ occurs.
If $t_1\geq2$, the statement is Theorem \ref{thm:t1-at-least-2}.  Assume
therefore $t_1=1$.  The comparison with the $\Z_4$-linear and the
$\Z_2\Z_4$-linear Hadamard codes is Proposition \ref{prop:t1-one-z2z4}.  For
the $\Z_{2^s}$-linear Hadamard codes, a linear comparison code is separated by
linearity, because $H^{1,t_2,t_3}$ is nonlinear, and a nonlinear one is
settled by Theorem \ref{thm:t1-one-z2s}, except when $t_2=1$ and the
comparison code is permutation equivalent to $\bar H^{2,0,t-5}$.  In that
remaining case, the $\Z_2\Z_4\Z_8$-linear Hadamard code is $H^{1,1,t_3}$ with, by
\eqref{eq:t1one-length}, $t_3=t-4$, so the pair is $H^{1,1,t-4}$ and
$\bar H^{2,0,t-5}$, which is settled by Corollary
\ref{cor:exceptional-family}; and a code permutation equivalent to
$\bar H^{2,0,t-5}$ is not equivalent to $H^{1,1,t-4}$ either.
\end{proof}

Combining Theorem \ref{thm:cross-family} with the internal classification of
\cite{Z2Z4Z8Profile}, the classification may be stated in its final form: for every
$t\geq3$, the $\lfloor(t^2+6)/12\rfloor$ codes $H^{t_1,t_2,t_3}$ of length
$2^t$ are pairwise nonequivalent, and none of them, apart from the unique
linear one, is equivalent to a $\Z_4$-linear, a $\Z_2\Z_4$-linear or a
$\Z_{2^s}$-linear Hadamard code of the same length, for any $s\geq2$.  

Table \ref{tab:families} summarises the three profiles side by side.  The
striking feature is the last column: the profile of the $\Z_2\Z_4\Z_8$-linear Hadamard codes is
nonconstant as soon as $t_1\geq2$, whereas the two comparison families are
always constant, and this single difference is what does most of the work
above.  The three families are moreover indistinguishable by the two coarser
data displayed in the third and fourth columns, since all of them have
$2^{\kappa-1}$ blocks of size $2^{\,t-\kappa+1}$; only the distribution of the
rank over those blocks differs.

\begin{table}[ht]
\small
\centering
\caption{The kernel-block rank profile of the three families of Hadamard codes
of length $2^t$ and kernel dimension $\kappa$.  For the $\Z_2\Z_4\Z_8$-linear Hadamard codes,
$\rho_0$ and $\rho_1$ are given by \eqref{eq:rho-def} and
$\rho_0-\rho_1=t_1-1$; for the $\Z_{2^s}$-linear family the value
$\rho_{a_1,\dots,a_s}=\rank(\bar H^{a_1,\dots,a_{s-1}})$ of
\eqref{eq:local-rank-general} is used, which is legitimate for every nonlinear
code by Remark \ref{rem:chain-reduction}.  The second column gives, for the
first two families, the exact criterion of nonlinearity; for the third one it
gives the criterion under which the codes are listed here, the general one
being read off \eqref{eq:sigma} and the linearity criterion recalled in
Subsection \ref{subsec:z2s-family}.}
\label{tab:families}
\begin{tabular}{@{}lcccc@{}}
\toprule
Code & nonlinear when & blocks & block size & profile\\
\midrule
$H^{t_1,t_2,t_3}$ & $(t_1,t_2)\not=(1,0)$ & $2^{\kappa-1}$
 & $2^{\,t-\kappa+1}$
 & $\bigl\{\!\bigl\{\rho_0^{\,[2^{\kappa-t_1-1}]},
    \rho_1^{\,[(2^{t_1}-1)2^{\kappa-t_1-1}]}\bigr\}\!\bigr\}$\\[1mm]
$H^{U,V}$ & $U\geq2$ & $2^{\kappa-1}$ & $2^{\,t-\kappa+1}$
 & $\bigl\{\!\bigl\{(U+1)^{\,[2^{\kappa-1}]}\bigr\}\!\bigr\}$\\[1mm]
$\bar H^{a_1,\dots,a_s}$ & $a_1\geq2$ ($a_1\geq3$ if $s=2$) & $2^{\kappa-1}$
 & $2^{\,t-\kappa+1}$
 & $\bigl\{\!\bigl\{\rho_{a_1,\dots,a_s}^{\,[2^{\kappa-1}]}\bigr\}\!\bigr\}$\\
\midrule
\multicolumn{5}{@{}l@{}}{The profile is constant for the last two families
always, and for the first one if and only if $t_1=1$.}\\
\bottomrule
\end{tabular}
\end{table}

Tables \ref{tab:small} and \ref{tab:large} list the type, the pair $(r,k)$
and the profile of every $\Z_4$-linear, $\Z_2\Z_4$-linear, $\Z_8$-linear,
$\Z_{2^4}$-linear, $\Z_{2^5}$-linear and $\Z_2\Z_4\Z_8$-linear Hadamard code
of length $2^t$, for $4\leq t\leq9$.  Four features deserve comment.  First,
the linear codes of each length are those with $r=k=t+1$, and they all have
the profile $\{\!\{1^{[2^t]}\}\!\}$ by Lemma \ref{lem:linear-profile}; the
interest of the tables lies in the nonlinear ones.  Second, the entries marked
with $\star$ agree in the length, the rank, the dimension of the kernel and
the whole profile, and only the two-block invariant $\mu$ separates the
$\Z_2\Z_4\Z_8$-linear one from the others; they are the members of the
exceptional family of Corollary \ref{cor:exceptional-family}.  Third, several
rows of the $\Z_{2^4}$-linear and $\Z_{2^5}$-linear families repeat the
invariants of a row above, which is Remark \ref{rem:chain-reduction} in
action: a $\Z_{2^s}$-linear Hadamard code with $a_1=1$ is permutation
equivalent to the $\Z_{2^{s-1}}$-linear one displayed there, and the two rows
are therefore the same code.  Fourth, seven codes have no entry in the rank
column, namely those that the same remark does not reduce below $s=4$, because
no formula for the rank of a $\Z_{2^s}$-linear Hadamard code is known for
$s\geq4$. Their profile, on the contrary, is available, because Corollary
\ref{cor:local-rank-bound} expresses it as the rank of a code one alphabet
lower.  This is the concrete form of the advantage claimed in the
introduction.

\begin{table}[!ht]
\scriptsize
\setlength{\tabcolsep}{3pt}
\caption{Type, rank $r$, dimension of the kernel $k$ and kernel-block rank profile $\cR$ of all the $\Z_4$-linear, $\Z_2\Z_4$-linear, $\Z_8$-linear, $\Z_{2^4}$-linear, $\Z_{2^5}$-linear and $\Z_2\Z_4\Z_8$-linear Hadamard codes of length $2^t$, for $4\leq t\leq6$.  By Remark \ref{rem:chain-reduction}, a $\Z_{2^s}$-linear Hadamard code with $a_1=1$ is permutation equivalent to a $\Z_{2^{s-1}}$-linear one, which is why some rows repeat the invariants of a row above.  The codes marked with $\star$ all share the length, the rank, the dimension of the kernel and the profile; they are the members of the exceptional family of Corollary \ref{cor:exceptional-family}, and the invariant $\mu$ of Definition \ref{def:mu} separates the $\Z_2\Z_4\Z_8$-linear one from the others.}
\label{tab:small}
\centering
\begin{tabular}{|c||ccc|ccc|ccc|}
\cline{1-10}
 & \multicolumn{3}{c|}{$t=4$}& \multicolumn{3}{c|}{$t=5$}& \multicolumn{3}{c|}{$t=6$}\\
\cline{2-10}
 & type & $(r,k)$ & $\cR$ & type & $(r,k)$ & $\cR$ & type & $(r,k)$ & $\cR$\\[0.2em]
\hline
\multirow{3}{*}{$\Z_4$} & $(2^{3}; 1,3)$ & $(5,5)$ & $\{\!\{1^{[2^{4}]}\}\!\}$ & $(2^{4}; 1,4)$ & $(6,6)$ & $\{\!\{1^{[2^{5}]}\}\!\}$ & $(2^{5}; 1,5)$ & $(7,7)$ & $\{\!\{1^{[2^{6}]}\}\!\}$\\
 & $(2^{3}; 2,1)$ & $(5,5)$ & $\{\!\{1^{[2^{4}]}\}\!\}$ & $(2^{4}; 2,2)$ & $(6,6)$ & $\{\!\{1^{[2^{5}]}\}\!\}$ & $(2^{5}; 2,3)$ & $(7,7)$ & $\{\!\{1^{[2^{6}]}\}\!\}$\\
 &  &  &  & $(2^{4}; 3,0)$ & $(7,4)$ & $\{\!\{3^{[8]}\}\!\}$ & $(2^{5}; 3,1)$ & $(8,5)$ & $\{\!\{3^{[2^{4}]}\}\!\}$\\[0.2em]
\hline
\multirow{3}{*}{$\Z_2\Z_4$} & $(8,4; 1,3)$ & $(5,5)$ & $\{\!\{1^{[2^{4}]}\}\!\}$ & $(16,8; 1,4)$ & $(6,6)$ & $\{\!\{1^{[2^{5}]}\}\!\}$ & $(32,16; 1,5)$ & $(7,7)$ & $\{\!\{1^{[2^{6}]}\}\!\}$\\
 & $(4,6; 2,1)$ & $(6,3)$ & $\{\!\{3^{[4]}\}\!\}$ & $(8,12; 2,2)$ & $(7,4)$ & $\{\!\{3^{[8]}\}\!\}$ & $(16,24; 2,3)$ & $(8,5)$ & $\{\!\{3^{[2^{4}]}\}\!\}$\\
 &  &  &  &  &  &  & $(8,28; 3,1)$ & $(10,4)$ & $\{\!\{4^{[8]}\}\!\}$\\[0.2em]
\hline
\multirow{4}{*}{$\Z_8$} & $(2^{2}; 1,0,2)$ & $(5,5)$ & $\{\!\{1^{[2^{4}]}\}\!\}$ & $(2^{3}; 1,0,3)$ & $(6,6)$ & $\{\!\{1^{[2^{5}]}\}\!\}$ & $(2^{4}; 1,0,4)$ & $(7,7)$ & $\{\!\{1^{[2^{6}]}\}\!\}$\\
 & $(2^{2}; 1,1,0)$ & $(5,5)$ & $\{\!\{1^{[2^{4}]}\}\!\}$ & $(2^{3}; 1,1,1)$ & $(6,6)$ & $\{\!\{1^{[2^{5}]}\}\!\}$ & $(2^{4}; 1,1,2)$ & $(7,7)$ & $\{\!\{1^{[2^{6}]}\}\!\}$\\
 &  &  &  & $(2^{3}; 2,0,0){}^{\star}$ & $(8,3)$ & $\{\!\{4^{[4]}\}\!\}$ & $(2^{4}; 1,2,0)$ & $(8,5)$ & $\{\!\{3^{[2^{4}]}\}\!\}$\\
 &  &  &  &  &  &  & $(2^{4}; 2,0,1){}^{\star}$ & $(9,4)$ & $\{\!\{4^{[8]}\}\!\}$\\[0.2em]
\hline
\multirow{3}{*}{$\Z_{2^4}$} & $(2; 1,0,0,1)$ & $(5,5)$ & $\{\!\{1^{[2^{4}]}\}\!\}$ & $(2^{2}; 1,0,0,2)$ & $(6,6)$ & $\{\!\{1^{[2^{5}]}\}\!\}$ & $(2^{3}; 1,0,0,3)$ & $(7,7)$ & $\{\!\{1^{[2^{6}]}\}\!\}$\\
 &  &  &  & $(2^{2}; 1,0,1,0)$ & $(6,6)$ & $\{\!\{1^{[2^{5}]}\}\!\}$ & $(2^{3}; 1,0,1,1)$ & $(7,7)$ & $\{\!\{1^{[2^{6}]}\}\!\}$\\
 &  &  &  &  &  &  & $(2^{3}; 1,1,0,0){}^{\star}$ & $(9,4)$ & $\{\!\{4^{[8]}\}\!\}$\\[0.2em]
\hline
\multirow{2}{*}{$\Z_{2^5}$} & $(1; 1,0,0,0,0)$ & $(5,5)$ & $\{\!\{1^{[2^{4}]}\}\!\}$ & $(2; 1,0,0,0,1)$ & $(6,6)$ & $\{\!\{1^{[2^{5}]}\}\!\}$ & $(2^{2}; 1,0,0,0,2)$ & $(7,7)$ & $\{\!\{1^{[2^{6}]}\}\!\}$\\
 &  &  &  &  &  &  & $(2^{2}; 1,0,0,1,0)$ & $(7,7)$ & $\{\!\{1^{[2^{6}]}\}\!\}$\\[0.2em]
\hline
\multirow{3}{*}{$\Z_2\Z_4\Z_8$} & $(4,2,2; 1,0,2)$ & $(5,5)$ & $\{\!\{1^{[2^{4}]}\}\!\}$ & $(8,4,4; 1,0,3)$ & $(6,6)$ & $\{\!\{1^{[2^{5}]}\}\!\}$ & $(16,8,8; 1,0,4)$ & $(7,7)$ & $\{\!\{1^{[2^{6}]}\}\!\}$\\
 &  &  &  & $(4,6,4; 1,1,1){}^{\star}$ & $(8,3)$ & $\{\!\{4^{[4]}\}\!\}$ & $(8,12,8; 1,1,2){}^{\star}$ & $(9,4)$ & $\{\!\{4^{[8]}\}\!\}$\\
 &  &  &  &  &  &  & $(4,6,12; 2,0,1)$ & $(12,3)$ & $\{\!\{6^{[1]},\,5^{[3]}\}\!\}$\\[0.2em]
\hline
\end{tabular}
\end{table}

\begin{table}[!ht]
\scriptsize
\setlength{\tabcolsep}{2pt}
\caption{The same data for $7\leq t\leq9$.  A dash in the column $r$ records that no formula for that rank is known. This happens exactly for the $\Z_{2^s}$-linear codes that Remark \ref{rem:chain-reduction} does not reduce below $s=4$, whereas their profile is given by Corollary \ref{cor:local-rank-bound}.  The single exception is the entry marked $\dagger$, where that corollary gives $\rho_{2,0,0,0,0}=\rank(\bar H^{2,0,0,0})$ and no formula for that rank is known; its value was obtained the value listed there was obtained with {\sc Magma}.}
\label{tab:large}
\centering
\resizebox{\textwidth}{!}{%
\begin{tabular}{|c||ccc|ccc|ccc|}
\cline{1-10}
 & \multicolumn{3}{c|}{$t=7$}& \multicolumn{3}{c|}{$t=8$}& \multicolumn{3}{c|}{$t=9$}\\
\cline{2-10}
 & type & $(r,k)$ & $\cR$ & type & $(r,k)$ & $\cR$ & type & $(r,k)$ & $\cR$\\[0.2em]
\hline
\multirow{5}{*}{$\Z_4$} & $(2^{6}; 1,6)$ & $(8,8)$ & $\{\!\{1^{[2^{7}]}\}\!\}$ & $(2^{7}; 1,7)$ & $(9,9)$ & $\{\!\{1^{[2^{8}]}\}\!\}$ & $(2^{8}; 1,8)$ & $(10,10)$ & $\{\!\{1^{[2^{9}]}\}\!\}$\\
 & $(2^{6}; 2,4)$ & $(8,8)$ & $\{\!\{1^{[2^{7}]}\}\!\}$ & $(2^{7}; 2,5)$ & $(9,9)$ & $\{\!\{1^{[2^{8}]}\}\!\}$ & $(2^{8}; 2,6)$ & $(10,10)$ & $\{\!\{1^{[2^{9}]}\}\!\}$\\
 & $(2^{6}; 3,2)$ & $(9,6)$ & $\{\!\{3^{[2^{5}]}\}\!\}$ & $(2^{7}; 3,3)$ & $(10,7)$ & $\{\!\{3^{[2^{6}]}\}\!\}$ & $(2^{8}; 3,4)$ & $(11,8)$ & $\{\!\{3^{[2^{7}]}\}\!\}$\\
 & $(2^{6}; 4,0)$ & $(11,5)$ & $\{\!\{4^{[2^{4}]}\}\!\}$ & $(2^{7}; 4,1)$ & $(12,6)$ & $\{\!\{4^{[2^{5}]}\}\!\}$ & $(2^{8}; 4,2)$ & $(13,7)$ & $\{\!\{4^{[2^{6}]}\}\!\}$\\
 &  &  &  &  &  &  & $(2^{8}; 5,0)$ & $(16,6)$ & $\{\!\{5^{[2^{5}]}\}\!\}$\\[0.2em]
\hline
\multirow{4}{*}{$\Z_2\Z_4$} & $(64,32; 1,6)$ & $(8,8)$ & $\{\!\{1^{[2^{7}]}\}\!\}$ & $(128,64; 1,7)$ & $(9,9)$ & $\{\!\{1^{[2^{8}]}\}\!\}$ & $(256,128; 1,8)$ & $(10,10)$ & $\{\!\{1^{[2^{9}]}\}\!\}$\\
 & $(32,48; 2,4)$ & $(9,6)$ & $\{\!\{3^{[2^{5}]}\}\!\}$ & $(64,96; 2,5)$ & $(10,7)$ & $\{\!\{3^{[2^{6}]}\}\!\}$ & $(128,192; 2,6)$ & $(11,8)$ & $\{\!\{3^{[2^{7}]}\}\!\}$\\
 & $(16,56; 3,2)$ & $(11,5)$ & $\{\!\{4^{[2^{4}]}\}\!\}$ & $(32,112; 3,3)$ & $(12,6)$ & $\{\!\{4^{[2^{5}]}\}\!\}$ & $(64,224; 3,4)$ & $(13,7)$ & $\{\!\{4^{[2^{6}]}\}\!\}$\\
 &  &  &  & $(16,120; 4,1)$ & $(15,5)$ & $\{\!\{5^{[2^{4}]}\}\!\}$ & $(32,240; 4,2)$ & $(16,6)$ & $\{\!\{5^{[2^{5}]}\}\!\}$\\[0.2em]
\hline
\multirow{8}{*}{$\Z_8$} & $(2^{5}; 1,0,5)$ & $(8,8)$ & $\{\!\{1^{[2^{7}]}\}\!\}$ & $(2^{6}; 1,0,6)$ & $(9,9)$ & $\{\!\{1^{[2^{8}]}\}\!\}$ & $(2^{7}; 1,0,7)$ & $(10,10)$ & $\{\!\{1^{[2^{9}]}\}\!\}$\\
 & $(2^{5}; 1,1,3)$ & $(8,8)$ & $\{\!\{1^{[2^{7}]}\}\!\}$ & $(2^{6}; 1,1,4)$ & $(9,9)$ & $\{\!\{1^{[2^{8}]}\}\!\}$ & $(2^{7}; 1,1,5)$ & $(10,10)$ & $\{\!\{1^{[2^{9}]}\}\!\}$\\
 & $(2^{5}; 1,2,1)$ & $(9,6)$ & $\{\!\{3^{[2^{5}]}\}\!\}$ & $(2^{6}; 1,2,2)$ & $(10,7)$ & $\{\!\{3^{[2^{6}]}\}\!\}$ & $(2^{7}; 1,2,3)$ & $(11,8)$ & $\{\!\{3^{[2^{7}]}\}\!\}$\\
 & $(2^{5}; 2,0,2){}^{\star}$ & $(10,5)$ & $\{\!\{4^{[2^{4}]}\}\!\}$ & $(2^{6}; 1,3,0)$ & $(12,6)$ & $\{\!\{4^{[2^{5}]}\}\!\}$ & $(2^{7}; 1,3,1)$ & $(13,7)$ & $\{\!\{4^{[2^{6}]}\}\!\}$\\
 & $(2^{5}; 2,1,0)$ & $(12,4)$ & $\{\!\{5^{[8]}\}\!\}$ & $(2^{6}; 2,0,3){}^{\star}$ & $(11,6)$ & $\{\!\{4^{[2^{5}]}\}\!\}$ & $(2^{7}; 2,0,4){}^{\star}$ & $(12,7)$ & $\{\!\{4^{[2^{6}]}\}\!\}$\\
 &  &  &  & $(2^{6}; 2,1,1)$ & $(13,5)$ & $\{\!\{5^{[2^{4}]}\}\!\}$ & $(2^{7}; 2,1,2)$ & $(14,6)$ & $\{\!\{5^{[2^{5}]}\}\!\}$\\
 &  &  &  & $(2^{6}; 3,0,0)$ & $(17,4)$ & $\{\!\{7^{[8]}\}\!\}$ & $(2^{7}; 2,2,0)$ & $(17,5)$ & $\{\!\{6^{[2^{4}]}\}\!\}$\\
 &  &  &  &  &  &  & $(2^{7}; 3,0,1)$ & $(18,5)$ & $\{\!\{7^{[2^{4}]}\}\!\}$\\[0.2em]
\hline
\multirow{9}{*}{$\Z_{2^4}$} & $(2^{4}; 1,0,0,4)$ & $(8,8)$ & $\{\!\{1^{[2^{7}]}\}\!\}$ & $(2^{5}; 1,0,0,5)$ & $(9,9)$ & $\{\!\{1^{[2^{8}]}\}\!\}$ & $(2^{6}; 1,0,0,6)$ & $(10,10)$ & $\{\!\{1^{[2^{9}]}\}\!\}$\\
 & $(2^{4}; 1,0,1,2)$ & $(8,8)$ & $\{\!\{1^{[2^{7}]}\}\!\}$ & $(2^{5}; 1,0,1,3)$ & $(9,9)$ & $\{\!\{1^{[2^{8}]}\}\!\}$ & $(2^{6}; 1,0,1,4)$ & $(10,10)$ & $\{\!\{1^{[2^{9}]}\}\!\}$\\
 & $(2^{4}; 1,0,2,0)$ & $(9,6)$ & $\{\!\{3^{[2^{5}]}\}\!\}$ & $(2^{5}; 1,0,2,1)$ & $(10,7)$ & $\{\!\{3^{[2^{6}]}\}\!\}$ & $(2^{6}; 1,0,2,2)$ & $(11,8)$ & $\{\!\{3^{[2^{7}]}\}\!\}$\\
 & $(2^{4}; 1,1,0,1){}^{\star}$ & $(10,5)$ & $\{\!\{4^{[2^{4}]}\}\!\}$ & $(2^{5}; 1,1,0,2){}^{\star}$ & $(11,6)$ & $\{\!\{4^{[2^{5}]}\}\!\}$ & $(2^{6}; 1,0,3,0)$ & $(13,7)$ & $\{\!\{4^{[2^{6}]}\}\!\}$\\
 & $(2^{4}; 2,0,0,0)$ & $(-,3)$ & $\{\!\{8^{[4]}\}\!\}$ & $(2^{5}; 1,1,1,0)$ & $(13,5)$ & $\{\!\{5^{[2^{4}]}\}\!\}$ & $(2^{6}; 1,1,0,3){}^{\star}$ & $(12,7)$ & $\{\!\{4^{[2^{6}]}\}\!\}$\\
 &  &  &  & $(2^{5}; 2,0,0,1)$ & $(-,4)$ & $\{\!\{8^{[8]}\}\!\}$ & $(2^{6}; 1,1,1,1)$ & $(14,6)$ & $\{\!\{5^{[2^{5}]}\}\!\}$\\
 &  &  &  &  &  &  & $(2^{6}; 1,2,0,0)$ & $(18,5)$ & $\{\!\{7^{[2^{4}]}\}\!\}$\\
 &  &  &  &  &  &  & $(2^{6}; 2,0,0,2)$ & $(-,5)$ & $\{\!\{8^{[2^{4}]}\}\!\}$\\
 &  &  &  &  &  &  & $(2^{6}; 2,0,1,0)$ & $(-,4)$ & $\{\!\{9^{[8]}\}\!\}$\\[0.2em]
\hline
\multirow{7}{*}{$\Z_{2^5}$} & $(2^{3}; 1,0,0,0,3)$ & $(8,8)$ & $\{\!\{1^{[2^{7}]}\}\!\}$ & $(2^{4}; 1,0,0,0,4)$ & $(9,9)$ & $\{\!\{1^{[2^{8}]}\}\!\}$ & $(2^{5}; 1,0,0,0,5)$ & $(10,10)$ & $\{\!\{1^{[2^{9}]}\}\!\}$\\
 & $(2^{3}; 1,0,0,1,1)$ & $(8,8)$ & $\{\!\{1^{[2^{7}]}\}\!\}$ & $(2^{4}; 1,0,0,1,2)$ & $(9,9)$ & $\{\!\{1^{[2^{8}]}\}\!\}$ & $(2^{5}; 1,0,0,1,3)$ & $(10,10)$ & $\{\!\{1^{[2^{9}]}\}\!\}$\\
 & $(2^{3}; 1,0,1,0,0){}^{\star}$ & $(10,5)$ & $\{\!\{4^{[2^{4}]}\}\!\}$ & $(2^{4}; 1,0,0,2,0)$ & $(10,7)$ & $\{\!\{3^{[2^{6}]}\}\!\}$ & $(2^{5}; 1,0,0,2,1)$ & $(11,8)$ & $\{\!\{3^{[2^{7}]}\}\!\}$\\
 &  &  &  & $(2^{4}; 1,0,1,0,1){}^{\star}$ & $(11,6)$ & $\{\!\{4^{[2^{5}]}\}\!\}$ & $(2^{5}; 1,0,1,0,2){}^{\star}$ & $(12,7)$ & $\{\!\{4^{[2^{6}]}\}\!\}$\\
 &  &  &  & $(2^{4}; 1,1,0,0,0)$ & $(-,4)$ & $\{\!\{8^{[8]}\}\!\}$ & $(2^{5}; 1,0,1,1,0)$ & $(14,6)$ & $\{\!\{5^{[2^{5}]}\}\!\}$\\
 &  &  &  &  &  &  & $(2^{5}; 1,1,0,0,1)$ & $(-,5)$ & $\{\!\{8^{[2^{4}]}\}\!\}$\\
 &  &  &  &  &  &  & $(2^{5}; 2,0,0,0,0)$ & $(-,3)$ & $\{\!\{14^{[4]}\}\!\}^{\dagger}$\\[0.2em]
\hline
\multirow{7}{*}{$\Z_2\Z_4\Z_8$} & $(32,16,16; 1,0,5)$ & $(8,8)$ & $\{\!\{1^{[2^{7}]}\}\!\}$ & $(64,32,32; 1,0,6)$ & $(9,9)$ & $\{\!\{1^{[2^{8}]}\}\!\}$ & $(128,64,64; 1,0,7)$ & $(10,10)$ & $\{\!\{1^{[2^{9}]}\}\!\}$\\
 & $(16,24,16; 1,1,3){}^{\star}$ & $(10,5)$ & $\{\!\{4^{[2^{4}]}\}\!\}$ & $(32,48,32; 1,1,4){}^{\star}$ & $(11,6)$ & $\{\!\{4^{[2^{5}]}\}\!\}$ & $(64,96,64; 1,1,5){}^{\star}$ & $(12,7)$ & $\{\!\{4^{[2^{6}]}\}\!\}$\\
 & $(8,28,16; 1,2,1)$ & $(13,4)$ & $\{\!\{5^{[8]}\}\!\}$ & $(16,56,32; 1,2,2)$ & $(14,5)$ & $\{\!\{5^{[2^{4}]}\}\!\}$ & $(32,112,64; 1,2,3)$ & $(15,6)$ & $\{\!\{5^{[2^{5}]}\}\!\}$\\
 & $(8,12,24; 2,0,2)$ & $(13,4)$ & $\{\!\{6^{[2]},\,5^{[6]}\}\!\}$ & $(16,24,48; 2,0,3)$ & $(14,5)$ & $\{\!\{6^{[4]},\,5^{[12]}\}\!\}$ & $(16,120,64; 1,3,1)$ & $(19,5)$ & $\{\!\{6^{[2^{4}]}\}\!\}$\\
 &  &  &  & $(8,28,48; 2,1,1)$ & $(19,4)$ & $\{\!\{7^{[2]},\,6^{[6]}\}\!\}$ & $(32,48,96; 2,0,4)$ & $(15,6)$ & $\{\!\{6^{[8]},\,5^{[24]}\}\!\}$\\
 &  &  &  &  &  &  & $(16,56,96; 2,1,2)$ & $(20,5)$ & $\{\!\{7^{[4]},\,6^{[12]}\}\!\}$\\
 &  &  &  &  &  &  & $(8,28,112; 3,0,1)$ & $(26,4)$ & $\{\!\{10^{[1]},\,8^{[7]}\}\!\}$\\[0.2em]
\hline
\end{tabular}
}
\end{table}

We close the section by recording explicitly what happens on the six infinite
families of pairs, listed in Proposition \ref{prop:recalled-cross-collisions},
that the rank and the dimension of the kernel leave undecided.

\begin{corollary}\label{cor:displayed-pairs}
The six infinite families of pairs of Proposition
\ref{prop:recalled-cross-collisions} consist of nonequivalent codes, for every
indicated value of $t$:
$$
\begin{array}{llll}
H^{2,1,t-7} &\text{and}& H^{5,t-9}, & t\geq10;\\[1mm]
H^{1,1,t-4} &\text{and}& \bar H^{2,0,t-5}, & t\geq5;\\[1mm]
H^{2,1,t-7} &\text{and}& \bar H^{1,5,t-12}, & t\geq12;\\[1mm]
H^{1,5,t-12} &\text{and}& \bar H^{4,0,t-11}, & t\geq13;\\[1mm]
H^{2,6,t-17} &\text{and}& \bar H^{4,3,t-17}, & t\geq18;\\[1mm]
H^{4,5,t-21} &\text{and}& \bar H^{6,2,t-21}, & t\geq22 .
\end{array}
$$
Only the second pair requires the invariant $\mu$; all the others are
separated by the kernel-block rank profile alone.
\end{corollary}

\begin{proof}
The $\Z_2\Z_4\Z_8$-linear Hadamard code of the first, third, fifth and sixth pairs has $t_1\geq2$, so
Theorem \ref{thm:t1-at-least-2} applies to each of them.  For the first one,
for instance, the profile of $H^{2,1,t-7}$ is
$\{\!\{7^{[2^{t-7}]},6^{[3\cdot2^{t-7}]}\}\!\}$ by Theorem \ref{thm:profile},
since $\rho_0=1+\binom{4}{2}=7$ and $\rho_1=1+2+\binom{3}{2}=6$, whereas that
of $H^{5,t-9}$ is constant, equal to $\{\!\{6^{[2^{t-5}]}\}\!\}$ by
Proposition \ref{prop:z2z4-profile}. The first multiset is nonconstant and the
second one is not, so they differ.

The fourth pair has $t_1=1$, so both profiles are constant and we compare
their values, as in Theorem \ref{thm:t1-one-z2s}.  Both codes have kernel
dimension $t-6$, hence $2^{\,t-7}$ kernel blocks of $2^{7}$ coordinates each;
but the constant value is $t_2+3=8$ for $H^{1,5,t-12}$ by
\eqref{eq:profile-t1-one}, and $0+1+\binom{5}{2}=11$ for $\bar H^{4,0,t-11}$
by Corollary \ref{cor:z8-local-rank}.  The two profiles are therefore
different.

Finally, the second pair is Corollary \ref{cor:exceptional-family}; it is the
only one of the six in which the two codes have the same profile, and
therefore the only one for which $\mu$ is needed.
\end{proof}
\section{Conclusions and further research}\label{sec:conclusions}

We have computed the kernel-block rank profile of the two families of Hadamard
codes with which the $\Z_2\Z_4\Z_8$-linear ones have to be compared, and we
have used it to complete their classification.

For the $\Z_2\Z_4$-linear Hadamard codes the computation is direct: a
nonlinear $H^{U,V}$ has $2^{U+V-1}$ kernel blocks, all of size $2^{U}$, and
constant profile of value $U+1$.  For the $\Z_{2^s}$-linear Hadamard codes $\bar H^{a_1,\dots,a_s}$ we proved, in Theorem \ref{thm:z2s-profile-constant}, that the profile is constant for every $s\geq2$.  Its proof produces coordinate permutations acting transitively on the kernel blocks: they translate the additive coordinates and
shift the Gray coordinates, the number of admissible shifts being governed by
the parameter $\sigma$ of \eqref{eq:sigma}.  No rank formula is used, which is
 why the argument is uniform in $s$. Table \ref{tab:large} shows three concrete codes for which the rank is unavailable while the profile is not.

The second main result is the descent Theorem \ref{thm:descent}, which
identifies the code punctured on a kernel block of a nonlinear
$\bar H^{a_1,\dots,a_s}$ with $a_1\geq2$ as the $\Z_{2^{s-1}}$-linear Hadamard
code $\bar H^{a_1,\dots,a_{s-1}}$.  It gives the constant local rank as
$\rank(\bar H^{a_1,\dots,a_{s-1}})$ and the sharp lower bound
$\rho_{a_1,\dots,a_s}\geq t-\kappa+2$, with equality exactly when $s=2$, or
$s=3$ and $a_1=2$.  Since the codes $H^{1,t_2,t_3}$ attain the same minimum,
only those two extremal situations remain, and the rank settles them except
for the family $H^{1,1,t-4}$ and $\bar H^{2,0,t-5}$ with $t\geq5$, whose
members share the length, the rank, the dimension of the kernel and the whole
profile, but are separated by the two-block invariant $\mu$ of Definition \ref{def:mu}
separates.  The outcome is Theorem \ref{thm:cross-family}, and it answers the
question stated as further research in \cite[\S 6]{Z2Z4Z8Linearity}, without
any computer equivalence test.  Ultimately, the families are separated by how much their local ranks exceed the lower bound $t-\kappa+2$, which is valid
for all three families. For the $\Z_2\Z_4\Z_8$-linear Hadamard codes, that excess takes the two values $\binom{t_1-1}{2}$ and $\binom{t_1-1}{2}+t_1-1$, by Lemma \ref{lem:mixed-local-rank-bound}; for the $\Z_2\Z_4$-linear codes it is always zero, since $t-\kappa+2=U+1$ there; and for a $\Z_{2^s}$-linear code it is the same on every block and measures how far the Hadamard code one alphabet lower is from being linear.

Several natural questions remain.  The first one concerns the rank of the
$\Z_{2^s}$-linear Hadamard codes.  Corollary \ref{cor:local-rank-bound}
expresses a local invariant of $\bar H^{a_1,\dots,a_s}$ as a global invariant
of a $\Z_{2^{s-1}}$-linear Hadamard code.  Read in the opposite direction, it
says that the rank of $\bar H^{a_1,\dots,a_{s-1}}$, which is known only for
$s-1\leq3$ \cite{Kro:2001:Z4_Had_Perf,PRV06,fernandez2019mathbb} and remains open
for $s-1\geq4$, is a local quantity of the code one alphabet higher,
computable on a single kernel block, which is much smaller than the whole
code.  The descent of Theorem \ref{thm:descent} thus reduces that open problem
to understanding how the rank grows when passing from a block to the whole
code.

The second question concerns ranks on unions of several blocks.  The invariant
$\mu$ of Definition \ref{def:mu} is the second level of a natural hierarchy:
for every $\ell\geq1$, we may consider the multiset of the dimensions
$\dim\langle C\rangle|_{B_1\cup\cdots\cup B_\ell}$, taken over all the
$\ell$-subsets of kernel blocks.  For $\ell=1$, this is the kernel-block rank
profile, and for $\ell=2$ it already separates the exceptional family of
Corollary \ref{cor:exceptional-family}, which the case $\ell=1$ cannot do.
It would be interesting to determine how much of the structure of an
arbitrary additive code is recovered by the whole hierarchy, and, in
particular, whether some $\ell$ always suffices to classify a given family.

The third question concerns other realisations and other alphabets.  As
recalled in \cite{Z2Z4Z8Profile}, the triple $(t_1,t_2,t_3)$ does not determine, up
to equivalence, every $\Z_2\Z_4\Z_8$-additive Hadamard code of that abstract
type, so all the statements of the present paper concern the recursively
constructed family of \cite{Z2Z4Z8Construction}, together with the
constructions already proved there to give permutation equivalent codes;
classifying all the $\Z_2\Z_4\Z_8$-additive Hadamard codes of a given abstract
type remains open.  It would also be natural to generalise the construction
and the present comparison to $\Z_2\Z_4\cdots\Z_{2^s}$-linear Hadamard codes
with all the $\alpha_i$ nonzero, or even to
$\Z_p\Z_{p^2}\cdots\Z_{p^s}$-linear generalized Hadamard codes with $p$ prime,
in the spirit of
\cite{HadamardZps,ZpZp2Construction,ZpZp2Classification}.  The kernel-block
rank profile extends verbatim to that generality and Lemma
\ref{lem:profile-invariant} holds without any change, so the invariant is
available. What remains to be done is to compute the labels and of the
local ranks.  In the same direction, the existence of $\Z_4\Z_8$-additive
Hadamard codes, that is, the case $\alpha_1=0$, $\alpha_2\not=0$,
$\alpha_3\not=0$, is still open, whereas the case $\alpha_1\not=0$,
$\alpha_2=0$, $\alpha_3\not=0$ cannot occur \cite{Z2Z4Z8Linearity}.

\bibliographystyle{elsarticle-num}
\bibliography{manuscript}

\end{document}